\pdfoutput=1
\documentclass[lettersize,journal]{IEEEtran}
\usepackage{amsmath,amsfonts}
\usepackage{algorithmic}
\usepackage{algorithm}
\usepackage{array}
\usepackage[caption=false,font=normalsize,labelfont=sf,textfont=sf]{subfig}
\usepackage{textcomp}
\usepackage{stfloats}
\usepackage{url}
\usepackage{verbatim}
\usepackage{graphicx}
\usepackage{cite}
\usepackage{bm}
\usepackage{mathtools}
\usepackage{amssymb}
\usepackage{amsthm}

\newtheorem{proposition}{Proposition}
\newtheorem{assumption}{Assumption}
\newtheorem{remark}{Remark}

\DeclareMathOperator{\Tr}{Tr}
\DeclareMathOperator{\diag}{diag}

\usepackage{threeparttable}
\usepackage{booktabs}  
\usepackage{makecell}   
\usepackage{multirow}  
\usepackage{enumerate}

\usepackage{caption}
\usepackage{stfloats}
\usepackage{subeqnarray}

\usepackage{cases}
\begin{document}

\title{Multicast Capacity of XL-RIS Assisted Hybrid Near- and Far-Field mmWave Communications}
\author{Hui Chen\textsuperscript{1}, Qi Wu\textsuperscript{2}, Hongcheng Zhuang\textsuperscript{1,2}\\
	{\normalsize \textsuperscript{1}School of Electronics and Communication Engineering,} 
	{\normalsize Sun Yat-sen University, Shenzhen 518107, China}\\
	{\normalsize \textsuperscript{2}Department of Broadband Communication, Peng Cheng Laboratory, Shenzhen 518055, China}\\
	{\normalsize chenh525@mail2.sysu.edu.cn, wuq05@pcl.ac.cn, zhuanghch@mail.sysu.edu.cn}
	\thanks{This work has been submitted to the IEEE for possible publication.  Copyright may be transferred without notice, after which this version may no longer be accessible.}
}
\maketitle

\begin{abstract}
 Multicast transmission in millimeter-wave (mmWave) networks is fundamentally limited by the weakest user, and blockages further exacerbate this problem. Large-scale reconfigurable intelligent surfaces (XL-RIS) offer a promising solution by providing high array gain to overcome blockages. However, the large aperture of XL-RIS significantly expands the near-field region, creating a hybrid-field scenario where some users lie in the near-field while others remain in the far-field. Existing hybrid-field studies on XL-RIS have primarily focused on channel estimation and deployment optimization, leaving multicast capacity analysis unexplored. This paper investigates the fundamental capacity limits of XL-RIS-assisted multicast communications in hybrid-field scenarios. For the fundamental two-user case consisting of one near-field and one far-field user, we derive the optimal closed-form covariance matrix and optimize the RIS phase shifts via manifold optimization. We establish that the multicast capacity scales as $\Theta(\log_2(MN))$ as the number of transmit antennas $M$ and/or RIS elements $N$ grow large, and prove this scaling is order-tight. Numerical results validate the bounds and show the impact of $M$, $N$, and distance on the multicast rate.
\end{abstract}


\section{Introduction}
Multicast transmission is a core technology in scenarios such as vehicular networks and real-time video streaming, where identical information must be delivered to multiple users simultaneously\cite{intro1,intro2}. As one of the essential transmission modes in 6G networks, multicast imposes stringent requirements on data rate and coverage capability\cite{intro3}. To meet these demands, large-scale antenna arrays and millimeter-wave (mmWave) bands have become critical enablers for 6G\cite{intro4,intro5}. However, mmWave signals are highly susceptible to blockage, and in multicast systems, the performance is limited by the weakest user, making unreliable connectivity become the primary performance bottleneck.

Large-scale reconfigurable intelligent surfaces (XL-RIS) have emerged as a promising solution to overcome blockage in mmWave communications \cite{intro8}. Compared with conventional RIS, the larger aperture of XL-RIS provides significantly higher array gain, which is particularly beneficial for multicast systems where the weakest user limits the overall data rate. By shaping the wireless environment through programmable reflection, an XL-RIS can provide alternative reflection paths for users whose direct links to the base station (BS) are obstructed. This capability improves coverage and spectral efficiency\cite{intro8-1,intro8-2}, particularly in dense urban scenarios where blockages are frequent. Moreover, XL-RIS operates passively with minimal power consumption, making it an energy-efficient solution for beyond-5G networks. 

With the large aperture of XL-RIS and the short wavelengths in mmWave bands, the near-field region expands significantly, creating a hybrid-field scenario where some users lie in the near-field while others remain in the far-field region of the RIS. This coexistence introduces new challenges for channel modeling, RIS phase design, and capacity analysis, particularly in multicast where the weakest user determines the data rate.
Recent studies have investigated the hybrid-field properties of XL-RIS\cite{intro8-3,intro8-4,intro8-5,intro8-6,intro8-7}. In \cite{intro8-3}, the authors consider a hybrid-field scenario where the BS–RIS link is far-field, while the RIS–user link is near-field with visual region effects, and develops a sparse Bayesian learning algorithm for channel estimation. The authors in \cite{intro8-4} focus on RIS element failures in a similar hybrid-field setting and propose a three-stage signal-assisted sparse representation-based channel estimation scheme with double domain filtering and failure-aware orthogonal matching pursuit algorithms to achieve robust channel estimation with low pilot overhead. Building on these, the authors in \cite{intro8-5} extend the study to three channel categories: far-far, far-near, and near-near fields. They propose a unified channel estimation method for XL-RIS assisted XL-MIMO systems with hybrid beamforming, exploiting the low-rank structure of the BS–RIS channel to achieve high accuracy with reduced pilot overhead. \cite{intro8-6} proposes a convolutional dictionary learning-based method for hybrid-field scenarios, where far-field and near-field paths of RIS-User coexist, by unrolling proximal gradient descent into neural networks to learn channel representations. Different from the aforementioned works that focus on channel estimation, \cite{intro8-7} investigates XL-RIS placement strategies for beam focusing in hybrid near-field and far-field mmWave communications. Although these works have advanced the understanding of XL-RIS in hybrid-field scenarios, they primarily focus on channel estimation and RIS deployment optimization, rather than multicast capacity analysis.
 
On the other hand, from the multicast capacity perspective, several studies have been conducted that do not consider the hybrid-field effect of XL-RIS. For instance, the authors in \cite{intro9-pre} focus on improving multicast capacity in ultra-dense low earth orbit network through beamforming design and subchannel assignment. For conventional far-field multicast systems, the authors in \cite{intro9} derive the fundamental capacity limits and characterize the asymptotic scaling behavior under different regimes of antennas and users. \cite{intro10} extends the asymptotic capacity analysis to antenna subset selection in multicast channels. More recently, RIS-assisted multicast capacity has been investigated under far-field assumptions, with joint optimization of transmit covariance and RIS phase shifts\cite{intro11}. They provide asymptotic maximal capacity order when the numbers of antennas, reflecting elements, or users grow large. Moreover, \cite{intro12} investigates the channel capacity of near-field multiuser communications, including the multicast channel as a special case. For the two-user near-field multicast scenario, they derive the optimal beamforming vector and the corresponding multicast capacity in closed form. They also provide asymptotic results showing that the near-field multicast capacity converges to a finite value as the number of antennas grows large, in contrast to the unbounded growth predicted by far-field models. Despite these efforts, the capacity analysis of multicast systems in the hybrid-field scenario, where near-field and far-field users coexist on the RIS–user link, has not been addressed, to the best of our knowledge.

To fill this gap, this paper investigates the multicast capacity in XL-RIS assisted hybrid-field mmWave communications. The main contributions are summarized as follows.

\begin{itemize}
	\item For the fundamental two-user scenario comprising one near-field and one far-field user of XL-RIS, we derive the optimal closed-form expression of the covariance matrix. Moreover, we propose a manifold optimization algorithm to efficiently optimize the RIS phase shifts, achieving near-optimal performance.
	
	\item We derive a rigorous upper bound on the multicast capacity for the two-user hybrid near-/far-field case, showing that it scales as $O(\log_2(MN))$ as $M$ and/or $N$ grow large. Under two mild geometric assumptions and the standard large-aperture stationary-phase approximation for the dominant RIS-side oscillatory sum, we further establish a matching lower bound, thereby showing the order-tight scaling $C=\Theta(\log_2(MN))$.
	
	\item Numerical results validate the tightness of the upper bound and to illustrate the key system behaviors, including the impact of the number of transmit antenna, the distance of user and RIS elements' number on the achievable multicast rate.

\end{itemize}

\section{SYSTEM MODEL}
\begin{figure}[htbp]
		\centering
		\includegraphics[width=0.40\textwidth]{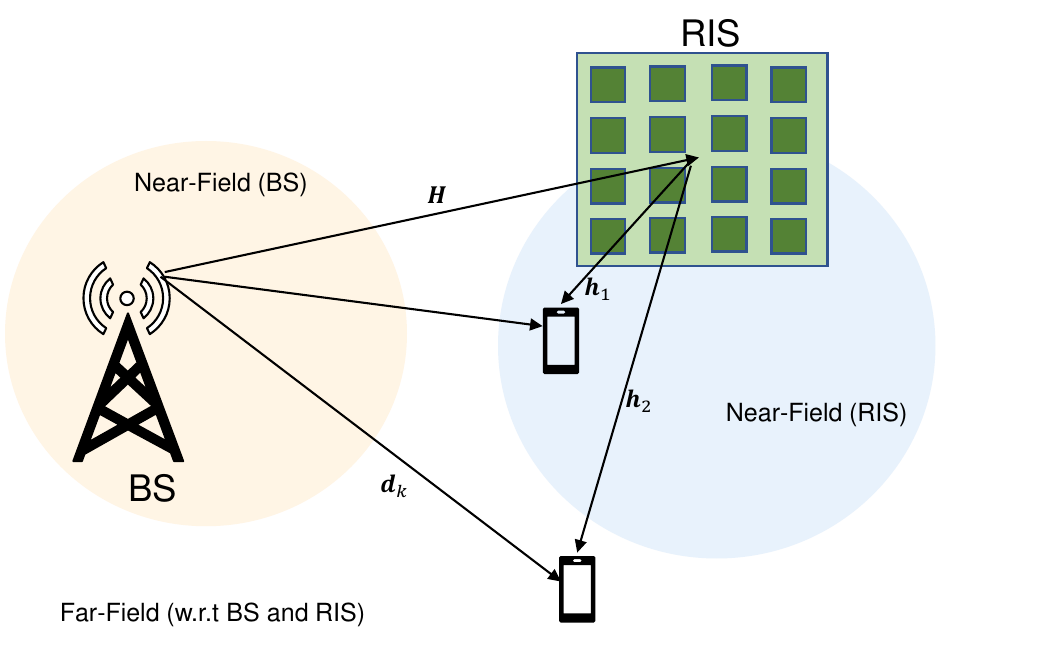}
	\caption{RIS-assisted multicast transmissions.}
	\label{fig1}
\end{figure}
We consider a downlink multicast communication system as illustrated in Fig. 1, where a BS equipped with an uniform planar array (UPA) of $M=M_hM_v$ antennas serves multiple single-antenna users with a common message. To overcome the obstruction of obstacles, an XL-RIS consisting of $N=N_hN_v$ passive reflecting elements arranged as a UPA is deployed to assist the communication.

Due to the large physical apertures of the XL-RIS, as well as the use of mmWave carrier frequencies, the conventional far-field assumption may no longer hold for all users. Specifically, users located within the near-field region of the XL-RIS experience spherical wavefronts and thus reside in the near-field, where the channel response depends on both angular direction and distance. In contrast, users outside these regions operate under the far-field condition, characterized by planar wavefronts and angle-only dependence.

\subsection{Channel Model}
The accurate modeling of the channels in the BS-RIS-user or BS-user link necessitates a clear demarcation between the near-field and far-field regimes, a distinction rigorously defined by the Rayleigh distance\cite{intro2-2}. Based on the relative positions of the nodes, we employ the appropriate channel model as follows.
\subsubsection{BS-RIS link}
First, the Rayleigh distance for the BS-RIS link is given by 
\begin{equation}\label{eq1}
	z_{BR}=\frac{2(D_{BS}+D_{RIS})^2}{\lambda}
\end{equation}
where the effective aperture size is a combination of both the BS and RIS dimensions, $D_{BS}=(\lambda/2)\sqrt{M_h^2+M_v^2}$ and $D_{RIS}=(\lambda/2)\sqrt{N_h^2+N_v^2}$, respectively.

While the link distance $z_H$ exceeds the threshold, i.e., $z_H > z_{BR}$, the channel operates in the far-field region. Consequently, the channel matrix $\boldsymbol{H} \in \mathbb{C}^{N \times M}$ can be modeled using planar-wavefront, which can be modeled as Saleh-Valenzuela \cite{channelSV} and can be expressed as:
\begin{equation}\label{eq2}
	\boldsymbol{H}=\sqrt{\frac{MN}{L_{BR}}}\sum_{l=0}^{L_{BR}-1} \beta_{l} \boldsymbol{a}_{R,l}\left(\varphi_{l}^{\text{AoA}}, \vartheta_{l}^{\text{AoA}}\right) \boldsymbol{a}_{B,l}^{H}\left(\varphi_{l}^{\text{AoD}}, \vartheta_{l}^{\text{AoD}}\right),
\end{equation}
where $L_{BR}$ is the number of propagation paths, each characterized by a complex gain $\beta_l\sim\mathcal{CN}(0,10^{-0.1PL_l(d)})$. $\boldsymbol{a}_R(\cdot)$ and $\boldsymbol{a}_B(\cdot)$ are the far-field array response vectors at the RIS and the BS, respectively. Moreover, we denote by $\theta_l^{\text{AoA}}$ and $\theta_l^{\text{AoD}}$ the angle of arrival (AoA) and the angle of departure (AoD) of the $l$-th path, respectively. The elevation and azimuth angles satisfy $\vartheta_l^{\text{AoA}} \in (0, \pi)$ and $\varphi_l^{\text{AoA}} \in (0, \pi)$, respectively.  

Given that both the BS and the RIS are equipped with UPAs, their response to a signal arriving from an azimuth angle $\varphi$ and elevation angle $\vartheta$ follows the closed-form expression below
\begin{equation}\label{eq3}
	\begin{aligned}
		\boldsymbol{a}_{\text{UPA}}  \left(\varphi, \vartheta\right) 
		= & \frac{1}{\sqrt{N_hN_v}}\big[1, \ldots, e^{j \frac{2 \pi d}{\lambda}\left(i_h \sin\vartheta\sin\varphi+i_v \cos\vartheta\right)} \\
		& \ldots, e^{j \frac{2 \pi d}{\lambda}\left(\left(N_h-1\right) \sin\vartheta\sin\varphi+\left(N_v-1\right) \cos\vartheta\right)}\big]^{\mathrm{T}},
	\end{aligned}
\end{equation}
where $d$ and $\lambda$ denote the inter-antenna spacing and the signal wavelength, respectively. $N_h$ and $N_v$ represent the numbers of horizontal and vertical antennas, respectively. The antenna indices in each dimension are denoted by $i_h \in [0, N_h-1]$ and $i_v \in [0, N_v-1]$.

\subsubsection{RIS-User $k$ Link}
For the RIS-User $k$ link, the Rayleigh distance is given by $z_{RU}=\frac{2D_{RIS}^2}{\lambda}$, which depends only on the RIS aperture $D_{RIS}$. 
\begin{enumerate}[a)]
	\item near-field scenario: For a user within the near-field region of RIS, its location is within $z_{h_k} \leq z_{RU}$.
Assuming that the RIS is placed in the $y-z$ plane, with $N_y$ and $N_z$ are odd integers, the element indices can be symmetrically defined as $\mathcal{N}_y=\{-\frac{N_y-1}{2},\cdots,\frac{N_y-1}{2}\}$ and $\mathcal{N}_z=\{-\frac{N_z-1}{2},\cdots,\frac{N_z-1}{2}\}$. The position of the $(n_y,n_z)$-th element is thus $\boldsymbol{s}_{n_y,n_z}=[0,n_yd,n_zd]^T$. The $k$-th near-field user is located at $\boldsymbol{u}_k=[u_{k,x},u_{k,y},u_{k,z}]^T=[r_k\sin\vartheta_k\cos\varphi_k,r_k\sin\vartheta_k\sin\varphi_k,r_k\cos\vartheta_k]^T$, where $r_k$ represents the propagation distance between the array center and the $k$-th near-field user, whose direction is specified by the elevation angle $\vartheta_k\in(0,\pi)$ and azimuth angle $\varphi_k\in(0,\pi)$. Specifically, the distance between the $k$-th near-field user and the array element $(0,n_yd,n_zd)$ is denoted by
\begin{equation}\label{eq4}
	\begin{aligned}
		&r_{n_y, n_z, k}=\Vert \boldsymbol{u}_k-\boldsymbol{s}_{n_y,n_z}\Vert\\
		&=r_k\sqrt{1+n_y^2\frac{d^2}{r_k^2}+n_z^2\frac{d^2}{r_k^2}-2n_y\frac{d}{r_k}\Phi_k-2n_z\frac{d}{r_k}\Omega_k} ,
	\end{aligned}
\end{equation}
where $\Phi_k=\sin\vartheta_k\sin\varphi_k$, $\Omega_k=\cos\vartheta_k$ and $\Psi_k=\sin\vartheta_k\cos\varphi_k$. Denote $\epsilon_k=\frac{d}{r_k}$.
The channel from the $(n_y,n_z)$-th element to the $k$-th near-field user, which considers free-space path-loss, effective aperture and polarization mismatch, can be expressed as\cite{intro12}
\begin{equation}\label{eq5}
	h_{n_y,n_z,k} \approx
	\sqrt{\frac{A r_k |\Psi_k|}{4\pi r_{n_y,n_z,k}^{3}}}
	e^{-j\frac{2\pi}{\lambda}r_{n_y,n_z,k}},
\end{equation}
where $A$ is the area occupied by each element, i.e., $\sqrt{A}$ is the element size along both the $y$- and $z$-axes. According to \cite{ref1,ref3,ref3-1}, near-field channels exhibit sparse scattering and are dominated by LoS propagation. Hence, we consider only the LoS path for the near-field channels.

Thus, the cascade channel vector $\boldsymbol{h}_k \in \mathbb{C}^{N \times 1}$
\begin{equation}\label{eq6}
	\boldsymbol{h}_k= \left[h_{-\frac{N_y-1}{2}, -\frac{N_z-1}{2},k}, \ldots, h_{\frac{N_y-1}{2}, \frac{N_z-1}{2},k}\right]^T.
\end{equation}
\item far-field scenario: If the user is in the far-field, i.e., $z_{h_k} > z_{RU}$, the channel can be represented as
	\begin{equation} \label{eq7}     
		\boldsymbol{h}_k= \sqrt{N/L_{k,RU}}\sum\nolimits_{l=0}^{L_{k,RU}-1}\beta_{l,k} \boldsymbol{a}_{R,l}\left(\varphi_{k,l}^{\text{AoD}}, \vartheta_{k,l}^{\text{AoD}}\right).
	\end{equation}
\end{enumerate}

\subsubsection{BS-User $k$ Link}
Both users are assumed to be located in the far-field region of the BS. In the far-field case, i.e., $z_{d_k} > z_{BU}$, where $z_{BU}=\frac{2D_{BS}^2}{\lambda}$, the planner wavefront must be considered, yielding the channel vector $\boldsymbol{d}_k \in \mathbb{C}^{M \times 1}$ as
\begin{equation}\label{eq8}
	\boldsymbol{d}_k= \sqrt{M/L_{k,BU}}\sum\nolimits_{l=1}^{L_{k,BU}}\alpha_{l,k} \boldsymbol{a}_{B,l}\left(\varphi_{k,l}^{\text{AoD}}, \vartheta_{k,l}^{\text{AoD}}\right),
\end{equation}
where the NLOS paths are only considered due to the obstacle blockage.

\subsection{Problem Formulation}
The received signal at user $k$ is 
\begin{equation}\label{eq11}
	y_k=(\boldsymbol{h}_k^H\boldsymbol{\Phi} \boldsymbol{H}+\boldsymbol{d}_k^H)\boldsymbol{s}+n_k,
\end{equation}
where $\boldsymbol{s}\in\mathbb{C}^{M\times1}$ is the transmitted symbol vector and $\boldsymbol{\Phi}=\text{diag}\{\boldsymbol{\phi}\}=\text{diag}\{e^{j\theta_1},\cdots,e^{j\theta_N}\}\in\mathbb{C}^{N\times N}$ is the RIS phase shifts matrix. 

Denote $\boldsymbol{R}=\mathbb{E}\{\boldsymbol{s}\boldsymbol{s}^H\}$ be the covariance matrix of $\boldsymbol{s}$ and $\boldsymbol{F}_k=\operatorname{diag}(\boldsymbol{h}_k)^H \boldsymbol{H}$. The multicast capacity is expressed as
\begin{equation}\label{eq12}
	\begin{aligned}
		C= & \max _{\boldsymbol{\phi},\boldsymbol{R} \succeq 0, \operatorname{Tr}(\boldsymbol{R}) \leq P_{\max }} \min _{\forall k} \\
		& \log_2 \left[1+\frac{\left(\boldsymbol{\phi}^H \boldsymbol{F}_k+\boldsymbol{d}_k^H\right) \boldsymbol{R}\left(\boldsymbol{F}_k^H\boldsymbol{\phi}+\boldsymbol{d}_k\right)}{\sigma_k^2}\right].
	\end{aligned}
\end{equation}
By introducing a non-negative auxiliary variable $\gamma$, \eqref{eq12} can be rewritten as
\begin{equation}\notag
	\begin{aligned}
		\mathcal{P}_1:
		\max_{\boldsymbol{\phi},\boldsymbol{R},\gamma}\quad & \gamma\\
		\text { st. }
		&\frac{\left(\boldsymbol{\phi}^H \boldsymbol{F}_k+\boldsymbol{d}_k^H\right) \boldsymbol{R}\left(\boldsymbol{F}_k^H \boldsymbol{\phi}+\boldsymbol{d}_k\right)}{\sigma_k^2}\geq\gamma,\\
		&\operatorname{Tr}(\boldsymbol{R}) \leq P_{\operatorname{max}},\boldsymbol{R} \succeq \boldsymbol{0},\\
		&\left|\phi_n\right|= 1, \forall n.
	\end{aligned}
\end{equation}

\section{CHANNEL CAPACITY}
To investigate hybrid near-/far-field propagation while maintaining analytical tractability, we focus on a two-user multicast setup where user~1 is in the RIS near-field and user~2 in the RIS far-field. Both users lie in the BS far-field (see Section~II).

Specifically, we first propose an alternating optimization (AO) algorithm to jointly optimize the covariance matrix and RIS phase shifts for suboptimal multicast capacity. Then, we derive an upper bound on the multicast capacity for fixed RIS phase shifts and analyze its asymptotic behavior.
\subsection{Proposed AO algorithm}
Accordingly, $\mathcal{P}_1$ is transformed into 
\begin{equation}\notag
	\begin{aligned}
		\mathcal{P}_2:
		\max_{\boldsymbol{R},\boldsymbol{\phi},\gamma}\quad & \gamma\\
		\text { st. }
		&\boldsymbol{f}_1^H\boldsymbol{R}\boldsymbol{f}_1\geq\gamma,\boldsymbol{f}_2^H\boldsymbol{R}\boldsymbol{f}_2\geq\gamma,\\
		&\operatorname{Tr}(\boldsymbol{R}) \leq P_{\operatorname{max}},\boldsymbol{R} \succeq \boldsymbol{0},\\
		&\left|\phi_n\right|= 1, \forall n,
	\end{aligned}
\end{equation}
where $\boldsymbol{f}_1=\frac{\boldsymbol{F}_1^H \boldsymbol{\phi}+\boldsymbol{d}_1}{\sigma_1}$, $\boldsymbol{f}_2=\frac{\boldsymbol{F}_2^H \boldsymbol{\phi}+\boldsymbol{d}_2}{\sigma_2}$. 

\begin{proposition}\label{prop:fixedphi}
	For a fixed RIS phase vector $\boldsymbol{\phi}$, the optimal two-user fair
	SINR of Problem~$\mathcal{P}_2$ is given by
	\begin{equation}\label{eq:gamma_exact}
		\gamma^\star(\boldsymbol{\phi})
		=
		P_{\max}
		\frac{a(\boldsymbol{\phi})b(\boldsymbol{\phi})-|c(\boldsymbol{\phi})|^2}
		{a(\boldsymbol{\phi})+b(\boldsymbol{\phi})-2|c(\boldsymbol{\phi})|},
	\end{equation}
	where
	\begin{equation}\label{eq:abc}
		a(\boldsymbol{\phi}) = \|\boldsymbol{f}_1\|^2,\quad
		b(\boldsymbol{\phi}) = \|\boldsymbol{f}_2\|^2,\quad
		c(\boldsymbol{\phi}) = \boldsymbol{f}_1^H\boldsymbol{f}_2.
	\end{equation}
\end{proposition}

\noindent\textit{Proof:}
The proof follows from the KKT analysis of Problem $P_2$. For fixed $\boldsymbol{\phi}$, the two fairness constraints are simultaneously active at the optimum, which gives $\boldsymbol{f}_1^H \boldsymbol{R}^{\star} \boldsymbol{f}_1 = \boldsymbol{f}_2^H \boldsymbol{R}^{\star} \boldsymbol{f}_2 = \gamma^{\star}(\boldsymbol\phi)$. Solving the corresponding stationarity and complementary slackness conditions yields the expression above. Since the appendix proof below no longer relies on the exact-SINR formula, the main role of Proposition~1 is to characterize the fixed-phase optimum in closed form.

$\mathcal{P}_2$ is a joint optimization over $\boldsymbol{R}$ and $\boldsymbol{\phi}$. To tackle it, we adopt an AO approach. For a fixed $\boldsymbol{\phi}$ satisfying $\left|\phi_n\right|= 1, \forall n$, we first optimize over $\boldsymbol{R}$ and $\gamma$. The Lagrange function for this sub-problem is given by
\begin{equation}\label{eq13}
	\begin{aligned}
	&\mathcal{L}(\boldsymbol{R},\gamma,\lambda_1,\lambda_2,\mu,\boldsymbol{Z})\\
	=&\gamma+\lambda_1(\boldsymbol{f}_1^H\boldsymbol{R}\boldsymbol{f}_1-\gamma)+\lambda_2(\boldsymbol{f}_2^H\boldsymbol{R}\boldsymbol{f}_2-\gamma)\\
	&+\mu(P_{max}-\operatorname{Tr}(\boldsymbol{R}))+\operatorname{Tr}(\boldsymbol{ZR}),
	\end{aligned}
\end{equation}
where $\lambda_1\geq0$, $\lambda_2\geq0$, $\mu\geq0$, $\boldsymbol{Z}\succeq \boldsymbol{0}$ are Lagrange multipliers.

Since the sub-problem in $\mathcal{P}_2$ is convex and satisfies Slater's condition, strong duality holds and the optimal solution must satisfy the KKT conditions\cite{convex}. Taking the derivative of the Lagrangian in \eqref{eq13} with respect to the primal variables yields the following stationarity conditions
\begin{equation}\label{eq14}
	\frac{\partial \mathcal{L}}{\partial \boldsymbol{R}}=\lambda_1\boldsymbol{f}_1\boldsymbol{f}_1^H+\lambda_2\boldsymbol{f}_2\boldsymbol{f}_2^H-\mu\boldsymbol{I}+\boldsymbol{Z}=\boldsymbol{0},
\end{equation}
\begin{equation}\label{eq15}
	\frac{\partial \mathcal{L}}{\partial \gamma}=1-\lambda_1-\lambda_2=0.
\end{equation}
The complementary slackness conditions associated with the inequality constraints are given by
\begin{equation}\label{eq16}
	\lambda_1(\boldsymbol{f}_1^H\boldsymbol{R}\boldsymbol{f}_1-\gamma)=0,
\end{equation}
\begin{equation}\label{eq17}
	\lambda_2(\boldsymbol{f}_2^H\boldsymbol{R}\boldsymbol{f}_2-\gamma)=0,
\end{equation}
\begin{equation}\label{eq18}
	\mu(P_{max}-\operatorname{Tr}(\boldsymbol{R}))=0,
\end{equation}
\begin{equation}\label{eq19}
	\operatorname{Tr}(\boldsymbol{ZR})=0.
\end{equation}

\begin{proposition}\label{prop:kkt}
	$\lambda_1>0$, $\lambda_2>0$. 
\end{proposition}
\begin{proof}
	Please see Appendix A.
\end{proof}

Accordingly, \eqref{eq16} and \eqref{eq17} can be converted to
\begin{equation}\label{eq21}
	\begin{aligned}
	&\boldsymbol{f}_1^H\boldsymbol{R}\boldsymbol{f}_1=\gamma,\\
	&\boldsymbol{f}_2^H\boldsymbol{R}\boldsymbol{f}_2=\gamma,
	\end{aligned}
\end{equation}
which indicates that both users experience the same signal-to-interference-plus-noise ratio (SINR) at the optimal solution, as expected from the max‑min fairness criterion.

By right-multiplying by $\boldsymbol{R}$ and combining with \eqref{eq19}, we rewrite \eqref{eq14} as
\begin{equation}\label{eq22}
	\lambda_1\boldsymbol{f}_1\boldsymbol{f}_1^H\boldsymbol{R}+\lambda_2\boldsymbol{f}_2\boldsymbol{f}_2^H\boldsymbol{R}=\mu\boldsymbol{R}.
\end{equation}
We further transform \eqref{eq22}, then we have
\begin{equation}\label{eq23}
	\begin{aligned}
	&\lambda_1\boldsymbol{f}_1^H\boldsymbol{f}_1\boldsymbol{f}_1^H\boldsymbol{R}\boldsymbol{f}_1+\lambda_2\boldsymbol{f}_1^H\boldsymbol{f}_2\boldsymbol{f}_2^H\boldsymbol{R}\boldsymbol{f}_1=\mu\boldsymbol{f}_1^H\boldsymbol{R}\boldsymbol{f}_1,\\
	&\lambda_1\boldsymbol{f}_2^H\boldsymbol{f}_1\boldsymbol{f}_1^H\boldsymbol{R}\boldsymbol{f}_2+\lambda_2\boldsymbol{f}_2^H\boldsymbol{f}_2\boldsymbol{f}_2^H\boldsymbol{R}\boldsymbol{f}_2=\mu\boldsymbol{f}_2^H\boldsymbol{R}\boldsymbol{f}_2,\\
	&\lambda_1\boldsymbol{f}_1^H\boldsymbol{f}_1\boldsymbol{f}_1^H\boldsymbol{R}\boldsymbol{f}_2+\lambda_2\boldsymbol{f}_1^H\boldsymbol{f}_2\boldsymbol{f}_2^H\boldsymbol{R}\boldsymbol{f}_2=\mu\boldsymbol{f}_1^H\boldsymbol{R}\boldsymbol{f}_2,\\
	&\lambda_1\boldsymbol{f}_2^H\boldsymbol{f}_1\boldsymbol{f}_1^H\boldsymbol{R}\boldsymbol{f}_1+\lambda_2\boldsymbol{f}_2^H\boldsymbol{f}_2\boldsymbol{f}_2^H\boldsymbol{R}\boldsymbol{f}_1=\mu\boldsymbol{f}_2^H\boldsymbol{R}\boldsymbol{f}_1.
	\end{aligned}
\end{equation}
Let $c=\vert c\vert e^{j\delta_1}$, $x_{12}=\boldsymbol{f}_1^H\boldsymbol{R}\boldsymbol{f}_2=\vert x_{12}\vert e^{j\delta_2}$. With \eqref{eq21}, \eqref{eq23} is rewritten as
\begin{subequations}
	\begin{align}
		&\lambda_1a\gamma+\lambda_2cx_{12}^*=\mu\gamma,\label{eq24a}\\
		&\lambda_1c^*x_{12}+\lambda_2b\gamma=\mu\gamma,\label{eq24b}\\
		&\lambda_1ax_{12}+\lambda_2c\gamma=\mu x_{12},\label{eq24c}\\
		&\lambda_1c^*\gamma+\lambda_2bx_{12}^*=\mu x_{12}^*.\label{eq24d}
	\end{align}
\end{subequations}
We first rewrite \eqref{eq24a} and \eqref{eq24c} as 
\begin{subequations}\label{eq25}
	\begin{align}
		&\mu-\lambda_1a=\frac{\lambda_2\vert c\vert \vert x_{12}\vert e^{j(\delta_1-\delta_2)}}{\gamma},\label{eq25a}\\
		&\mu-\lambda_1a=\frac{\lambda_2\vert c\vert e^{j\delta_1}\gamma}{\vert x_{12}\vert e^{j\delta_2}}.\label{eq25b}
	\end{align}
\end{subequations}
Accordingly, we have
\begin{equation}\label{eq27}
		\vert x_{12}\vert^2=\gamma^2\rightarrow \vert x_{12}\vert=\gamma .
\end{equation}

According to the cyclic property of trace, combining \eqref{eq15} and \eqref{eq21}, \eqref{eq22} can be converted to
\begin{equation}\label{eq28}
	\begin{aligned}
		&\lambda_1\operatorname{Tr}(\boldsymbol{f}_1^H\boldsymbol{R}\boldsymbol{f}_1)+\lambda_2\operatorname{Tr}(\boldsymbol{f}_2^H\boldsymbol{R}\boldsymbol{f}_2)=\mu\operatorname{Tr}(\boldsymbol{R}),\\
		&\rightarrow\gamma=\mu P_{max}.
	\end{aligned}
\end{equation}
Obviously, $\gamma>0$, so $\mu=\frac{\gamma}{P_{max}}$ and $\mu$ is real number. Accordingly, from \eqref{eq25a} and Euler's formula, we have
\begin{equation}\label{eq29}
	\begin{aligned}
		&\mu = \lambda_1a+\lambda_2\vert c\vert  \cos(\delta_1-\delta_2),\\
	&\lambda_2\vert c\vert \sin(\delta_1-\delta_2)=0,\\
	&\rightarrow \delta_1-\delta_2=t\pi.
	\end{aligned}
\end{equation}
Then, $\mu = \lambda_1a+(-1)^t\lambda_2\vert c\vert $. Thus we have
\begin{equation}\label{eq30}
	\mu = \begin{cases}
		\lambda_1a-\lambda_2\vert c\vert , & \text{if } t \text{ is odd}, \\
		\lambda_1a+\lambda_2\vert c\vert, & \text{if } t \text{ is even}.
	\end{cases}
\end{equation}
Because $\mathcal{P}_1$ is a maximization problem, then $\gamma$ is need to be maximized, thus $\mu$ should be as large as possible. Accordingly, $t$ is even. Thus we obtain
\begin{equation}\label{eq31}
	\delta_1=\delta_2.
\end{equation}

According to \eqref{eq27} and \eqref{eq31}, \eqref{eq24a} and \eqref{eq24b} can be rewritten as
\begin{subequations}
	\begin{align}
		&\lambda_1a+\lambda_2\vert c\vert =\mu,\\
		&\lambda_1\vert c\vert+\lambda_2b=\mu.
	\end{align}
\end{subequations}
Combining with \eqref{eq15}, we can obtain
\begin{subequations}
	\begin{align}
		&\lambda_1^\star=\frac{b-\vert c\vert}{a+b-2\vert c\vert},\label{eq33a}\\
		&\lambda_2^\star=\frac{a-\vert c\vert}{a+b-2\vert c\vert},\label{eq33b}\\
		&\mu^\star=\frac{ab-\vert c\vert^2}{a+b-2\vert c\vert},\label{eq33c}
	\end{align}
\end{subequations}

Based on \eqref{eq28} and \eqref{eq33c}, we have
\begin{equation}
	\gamma^\star = \mu^\star P_{max}=\frac{ab-\vert c\vert^2}{a+b-2\vert c\vert}P_{max}.
\end{equation}

From \eqref{eq22}, we observe that the column space of $\boldsymbol{R}$ is contained in $\operatorname{span} \{\boldsymbol{f}_1,\boldsymbol{f}_2\}$. This follows from the fact that the left-hand side of \eqref{eq22} lies in $\operatorname{span} \{\boldsymbol{f}_1,\boldsymbol{f}_2\}$, and the right-hand side is a scaled version of $\boldsymbol{R}$. Consequently, the rank of $\boldsymbol{R}$ is at most two. Let $ \{\boldsymbol{u}_1,\boldsymbol{u}_2\}$ be an orthogonal basis for the two-dimensional subspace $\operatorname{span} \{\boldsymbol{f}_1,\boldsymbol{f}_2\}$, obtained via Gram-Schmidt orthogonalization, e.g.,
\begin{equation}
	\boldsymbol{u}_1=\frac{\boldsymbol{f}_1}{\Vert \boldsymbol{f}_1\Vert}, \boldsymbol{u}_2=\frac{\boldsymbol{f}_2-(\boldsymbol{f}_2^H\boldsymbol{u}_1)\boldsymbol{u}_1}{\Vert \boldsymbol{f}_2-(\boldsymbol{f}_2^H\boldsymbol{u}_1)\boldsymbol{u}_1\Vert}.
\end{equation}
Define the $M\times 2$ matrix $\boldsymbol{U}=[\boldsymbol{u}_1,\boldsymbol{u}_2]$, which satisfies $\boldsymbol{U}^H\boldsymbol{U}=\boldsymbol{I}_2$. Since the column space of $\boldsymbol{R}$ lies in $\operatorname{span} \{\boldsymbol{u}_1,\boldsymbol{u}_2\}$, there exists a $2\times 2$ Hermitian positive semi-definite matrix $\boldsymbol{Q}$ such that
\begin{equation}\label{eq36}
	\boldsymbol{R}=\boldsymbol{U}\boldsymbol{Q}\boldsymbol{U}^H,
\end{equation}
Assuming a Hermitian matrix $\boldsymbol{Q}=\left[\begin{array}{ll}q_{11} & q_{12} \\ q_{12}^* & q_{22}\end{array}\right]$. Substituting it into \eqref{eq36} yields
\begin{equation}\label{eq37}
	\boldsymbol{R}
	=q_{11}\boldsymbol{u}_1\boldsymbol{u}_1^H+q_{22}\boldsymbol{u}_2\boldsymbol{u}_2^H+q_{12}\boldsymbol{u}_1\boldsymbol{u}_2^H+q_{12}^*\boldsymbol{u}_2\boldsymbol{u}_1^H.
\end{equation}

Since $\{\boldsymbol{u}_1,\boldsymbol{u}_2\}$ forms an orthogonal basis for $\operatorname{span} \{\boldsymbol{f}_1,\boldsymbol{f}_2\}$, the vector $\boldsymbol{f}_2$ can be uniquely expressed as a linear combination of $\boldsymbol{u}_1$ and $\boldsymbol{u}_2$. Let 
\begin{equation}\label{eq38}
	\boldsymbol{f}_2=\alpha\boldsymbol{u}_1+\beta\boldsymbol{u}_2,
\end{equation}
where $\alpha$ and $\beta$ are complex coefficients. Taking the inner product with $\boldsymbol{u}_1$ and using orthogonality, we obtain
\begin{equation}\label{eq39}
	\alpha = \boldsymbol{u}_1^H\boldsymbol{f}_2-\beta\boldsymbol{u}_1^H\boldsymbol{u}_2=\frac{\boldsymbol{f}_1^H\boldsymbol{f}_2}{\Vert \boldsymbol{f}_1\Vert}=\frac{c}{\sqrt{a}}.
\end{equation}
The norm of $\boldsymbol{f}_2$ can be expressed as 
\begin{equation}\label{eq40}
	\Vert\boldsymbol{f}_2\Vert^2=\Vert\alpha\boldsymbol{u}_1\Vert^2+\Vert\beta\boldsymbol{u}_2\Vert^2,
\end{equation}
which, together with $\Vert \boldsymbol{f}_2\Vert^2=b$, yields $\beta^2=b-\frac{\vert c\vert^2}{a}$. Without loss of generality, we can take $\beta$ to be real and non-negative by absorbing any phase into $\boldsymbol{u}_2$ so that
\begin{equation} \label{eq41}
	\beta = \sqrt{b-\frac{\vert c\vert^2}{a}}.
\end{equation}
 
Based on \eqref{eq21}, \eqref{eq39} and \eqref{eq41}, we can obtain
\begin{subequations}\label{eq42}
	\begin{align}
		\gamma=&\boldsymbol{f}_1^H\boldsymbol{R}\boldsymbol{f}_1\notag\\
		=&q_{11}\boldsymbol{f}_1^H\boldsymbol{u}_1\boldsymbol{u}_1^H\boldsymbol{f}_1+q_{22}\boldsymbol{f}_1^H\boldsymbol{u}_2\boldsymbol{u}_2^H\boldsymbol{f}_1\notag\\
		&+q_{12}\boldsymbol{f}_1^H\boldsymbol{u}_1\boldsymbol{u}_2^H\boldsymbol{f}_1+q_{12}^*\boldsymbol{f}_1^H\boldsymbol{u}_2\boldsymbol{u}_1^H\boldsymbol{f}_1\notag\\
		=&q_{11}a,\label{eq42a}\\
		\gamma=&\boldsymbol{f}_2^H\boldsymbol{R}\boldsymbol{f}_2\notag\\
		=&q_{11}\boldsymbol{f}_2^H\boldsymbol{u}_1\boldsymbol{u}_1^H\boldsymbol{f}_2+q_{22}\boldsymbol{f}_2^H\boldsymbol{u}_2\boldsymbol{u}_2^H\boldsymbol{f}_2\notag\\
		&+q_{12}\boldsymbol{f}_2^H\boldsymbol{u}_1\boldsymbol{u}_2^H\boldsymbol{f}_2+q_{12}^*\boldsymbol{f}_2^H\boldsymbol{u}_2\boldsymbol{u}_1^H\boldsymbol{f}_2\notag\\
		=&q_{11}\vert \alpha\vert^2+q_{22}\vert\beta\vert^2+2\Re\{\alpha^*\beta q_{12}\},\label{eq42b}\\
		x_{12}=&\boldsymbol{f}_1^H\boldsymbol{R}\boldsymbol{f}_2\notag\\
		=&q_{11}\boldsymbol{f}_1^H\boldsymbol{u}_1\boldsymbol{u}_1^H\boldsymbol{f}_2+q_{22}\boldsymbol{f}_1^H\boldsymbol{u}_2\boldsymbol{u}_2^H\boldsymbol{f}_2\notag\\
		&+q_{12}\boldsymbol{f}_1^H\boldsymbol{u}_1\boldsymbol{u}_2^H\boldsymbol{f}_2+q_{12}^*\boldsymbol{f}_1^H\boldsymbol{u}_2\boldsymbol{u}_1^H\boldsymbol{f}_2\notag\\
		=&q_{11}\alpha\sqrt{a}+q_{12}\beta\sqrt{a}.\label{eq42c}
	\end{align}
\end{subequations}

\begin{equation}
	\begin{aligned}
	\gamma=&\boldsymbol{f}_1^H\boldsymbol{R}\boldsymbol{f}_1\\
	=&q_{11}\boldsymbol{f}_1^H\boldsymbol{u}_1\boldsymbol{u}_1^H\boldsymbol{f}_1+q_{22}\boldsymbol{f}_1^H\boldsymbol{u}_2\boldsymbol{u}_2^H\boldsymbol{f}_1\\
	&+q_{12}\boldsymbol{f}_1^H\boldsymbol{u}_1\boldsymbol{u}_2^H\boldsymbol{f}_1+q_{12}^*\boldsymbol{f}_1^H\boldsymbol{u}_2\boldsymbol{u}_1^H\boldsymbol{f}_1\\
	=&q_{11}a.
\end{aligned}
\end{equation}
According to \eqref{eq27} and \eqref{eq31}, we have $x_{12}=\gamma\frac{c}{\vert c\vert}$. Then, \eqref{eq42} can be written as
\begin{equation}\label{eq44}
	\begin{aligned}
		&q_{11}=\frac{\gamma}{a},\\
		&q_{11}\vert \alpha\vert^2+q_{22}\vert\beta\vert^2+2\Re\{\alpha^*\beta q_{12}\}=\gamma,\\
		&q_{11}\alpha\sqrt{a}+q_{12}\beta\sqrt{a}=\gamma\frac{c}{\vert c\vert}.
	\end{aligned}
\end{equation}
Solving this equation, we can obtain the solution
\begin{equation}\label{eq45}
	\begin{aligned}
	&q_{12}=\frac{\gamma}{\beta\sqrt{a}}(\frac{c}{\vert c\vert}-\frac{\alpha}{\sqrt{a}}),\\
	&q_{22}=\frac{\gamma}{\vert\beta\vert^2}\left[1-\frac{\vert\alpha\vert^2}{a}-2\Re\left\{\frac{\alpha^*}{\sqrt{a}}(\frac{c}{\vert c\vert}-\frac{\alpha}{\sqrt{a}})\right\}\right].
    \end{aligned}
\end{equation}
From \eqref{eq44} and \eqref{eq45}, we obtain the final expression for $\boldsymbol{R}$ below
\begin{equation}
	\begin{aligned}
	\boldsymbol{R}^\star&=\frac{\gamma^\star}{a}\boldsymbol{u}_1\boldsymbol{u}_1^H+2\Re\{\frac{\gamma^\star}{\beta\sqrt{a}}(\frac{c}{\vert c\vert}-\frac{c}{a})\boldsymbol{u}_1\boldsymbol{u}_2^H\}\\
	&+\frac{\gamma^\star}{\vert\beta\vert^2}\left[1-\frac{\vert\alpha\vert^2}{a}-2\Re\left\{\frac{\alpha^*}{\sqrt{a}}(\frac{c}{\vert c\vert}-\frac{\alpha}{\sqrt{a}})\right\}\right]\boldsymbol{u}_2\boldsymbol{u}_2^H
	.
\end{aligned}
\end{equation}

Accordingly, the sub-optimal capacity is 
\begin{equation}\label{eq47}
	\begin{aligned}
	&C^\star(\boldsymbol{\phi})=\log_2(1+\gamma)\\
	&=\log_2\left(1+P_{\max } \frac{a(\boldsymbol{\phi}) b(\boldsymbol{\phi})-|c(\boldsymbol{\phi})|^2}{a(\boldsymbol{\phi})+b(\boldsymbol{\phi})-2|c(\boldsymbol{\phi})|}\right),
\end{aligned}
\end{equation}

After obtaining the closed‑form expressions for $\boldsymbol{R}^\star$ and $\gamma^\star$ for a fixed RIS phase shifts, we now turn to optimizing $\boldsymbol{\phi}$ with given $\boldsymbol{R}$. Substituting the closed-form optimal $\boldsymbol{R}$ obtained in the previous subsection into the multicast capacity expression, the joint optimization over $\{\boldsymbol{R},\boldsymbol{\phi}\}$ reduces to the formulation over $\boldsymbol{\phi}$ only:
\begin{equation}\notag
		\mathcal{P}_3:
		\max_{\boldsymbol{\phi}}\quad  C(\boldsymbol{\phi})\quad\text { st. }
		\left|\phi_n\right|= 1, \forall n,
\end{equation}
where $C(\boldsymbol{\phi})$ is given by \eqref{eq47}. Since the constraint set is a Riemannian manifold, we adopt a manifold optimization approach\cite{alg1}.

Define the manifold:
\begin{equation}
	\mathcal{M}=\left\{\boldsymbol{\phi} \in \mathbb{C}^{N\times 1}:\left|\phi_n\right|=1, \forall n\right\}.
\end{equation}
Due to the monotonicity of the $\log_2(\cdot)$ function, the objective function $C(\boldsymbol{\phi})$ can be simplified as
\begin{equation}
	G(\boldsymbol{\phi})=P_{\max }  \frac{a(\boldsymbol{\phi}) b(\boldsymbol{\phi})-|c(\boldsymbol{\phi})|^2}{a(\boldsymbol{\phi})+b(\boldsymbol{\phi})-2|c(\boldsymbol{\phi})|},
\end{equation}
with $a(\boldsymbol{\phi})=\frac{\left\|\boldsymbol{F}_1^H \boldsymbol{\phi}+\boldsymbol{d}_1\right\|^2}{\sigma_1^2}$, $b(\boldsymbol{\phi})=\frac{\left\|\boldsymbol{F}_2^H \boldsymbol{\phi}+\boldsymbol{d}_2\right\|^2}{\sigma_2^2}$ and $c(\boldsymbol{\phi})=\frac{\left(\boldsymbol{F}_1^H \boldsymbol{\phi}+\boldsymbol{d}_1\right)^H\left(\boldsymbol{F}_2^H \boldsymbol{\phi}+\boldsymbol{d}_2\right)}{\sigma_1 \sigma_2}$. To apply Riemannian gradient ascent, we first compute the Euclidean gradient of $G(\boldsymbol{\phi})$ with respect to $\phi$. Using Wirtinger derivatives and the chain rule,
\begin{equation}
	\frac{\partial G}{\partial \boldsymbol{\phi}^*}=\frac{\partial G}{\partial a} \frac{\partial a}{\partial \boldsymbol{\phi}^*}+\frac{\partial G}{\partial b} \frac{\partial b}{\partial \boldsymbol{\phi}^*}+\frac{\partial G}{\partial c} \frac{\partial c}{\partial \boldsymbol{\phi}^*}+\frac{\partial G}{\partial c^*} \frac{\partial c^*}{\partial \boldsymbol{\phi}^*}.
\end{equation}
The partial derivatives are evaluated as
\begin{subequations}
	\begin{align}
		&\frac{\partial G}{\partial a}=P_{\max }  \frac{( b-|c|)^2}{(a+b-2|c|)^2},\\
		&\frac{\partial G}{\partial b}=P_{\max }  \frac{( a-|c|)^2}{(a+b-2|c|)^2},\\
		&\frac{\partial G}{\partial c}=\frac{\partial G}{\partial \vert c\vert}\frac{\partial \vert c\vert}{\partial c}=P_{\max }\frac{(a-\vert c\vert)(b-\vert c\vert)}{(a+b-2|c|)^2}\frac{c^*}{\vert c\vert},\\
		&\frac{\partial G}{\partial c^*}=\frac{\partial G}{\partial \vert c\vert}\frac{\partial \vert c\vert}{\partial c^*}=P_{\max }\frac{(a-\vert c\vert)(b-\vert c\vert)}{(a+b-2|c|)^2}\frac{c}{\vert c\vert}\\
		&\frac{\partial a}{\partial \boldsymbol{\phi}^*}=\frac{\boldsymbol{F}_1\boldsymbol{F}_1^H\boldsymbol{\phi}+\boldsymbol{F}_1\boldsymbol{d}_1}{\sigma_1^2},\\
		& \frac{\partial b}{\partial \boldsymbol{\phi}^*}=\frac{\boldsymbol{F}_2\boldsymbol{F}_2^H\boldsymbol{\phi}+\boldsymbol{F}_2\boldsymbol{d}_2}{\sigma_2^2},\\
		& \frac{\partial c}{\partial \boldsymbol{\phi}^*}=\frac{\boldsymbol{F}_1\boldsymbol{F}_2^H\boldsymbol{\phi}+\boldsymbol{F}_1\boldsymbol{d}_2}{\sigma_1\sigma_2},\\
		& \frac{\partial c^*}{\partial \boldsymbol{\phi}^*}=\frac{\boldsymbol{F}_2\boldsymbol{F}_1^H\boldsymbol{\phi}+\boldsymbol{F}_2\boldsymbol{d}_1}{\sigma_1\sigma_2}.
	\end{align}
\end{subequations}
Substituting these into the chain rule yields the Euclidean gradient $\nabla_{\boldsymbol{\phi}} G$:
\begin{equation}
	\begin{aligned}
	\nabla_{\boldsymbol{\phi}} G=&P_{\max }  \frac{( b-|c|)^2}{(a+b-2|c|)^2}\frac{\boldsymbol{F}_1\boldsymbol{F}_1^H\boldsymbol{\phi}+\boldsymbol{F}_1\boldsymbol{d}_1}{\sigma_1^2}\\
	&+P_{\max }  \frac{( a-|c|)^2}{(a+b-2|c|)^2}\frac{\boldsymbol{F}_2\boldsymbol{F}_2^H\boldsymbol{\phi}+\boldsymbol{F}_2\boldsymbol{d}_2}{\sigma_2^2}\\
	&+P_{\max }\frac{(a-\vert c\vert)(b-\vert c\vert)}{(a+b-2|c|)^2}\frac{c^*}{\vert c\vert}\frac{\boldsymbol{F}_1\boldsymbol{F}_2^H\boldsymbol{\phi}+\boldsymbol{F}_1\boldsymbol{d}_2}{\sigma_1\sigma_2}\\
	&+P_{\max }\frac{(a-\vert c\vert)(b-\vert c\vert)}{(a+b-2|c|)^2}\frac{c}{\vert c\vert}\frac{\boldsymbol{F}_2\boldsymbol{F}_1^H\boldsymbol{\phi}+\boldsymbol{F}_2\boldsymbol{d}_1}{\sigma_1\sigma_2}.
	\end{aligned}
\end{equation}

On the manifold $\mathcal{M}$, the Riemannian gradient is the projection of the Euclidean gradient onto the tangent space:
\begin{equation}
	\operatorname{grad}_{\mathcal{M}} G(\boldsymbol{\phi})=\nabla_{\boldsymbol{\phi}} G+\Re\{\nabla_{\boldsymbol{\phi}} G \odot \boldsymbol{\phi}^*\} \odot\boldsymbol{\phi},
\end{equation}
where $\odot$ denotes element-wise multiplication. The $n$-th component is 
\begin{equation}
	\left[\operatorname{grad}_{\mathcal{M}} G\right]_n=\left[\nabla_{\boldsymbol{\phi}} G\right]_n+\Re\left(\left[\nabla_{\boldsymbol{\phi}} G\right]_n \cdot \phi_n^*\right) \cdot \phi_n.
\end{equation}

Finally, we update $\boldsymbol{\phi}$ using Riemannian gradient ascent with step size $\tau$ and the retraction $R_{\boldsymbol{\phi}}(\boldsymbol{\nu})=\frac{\boldsymbol{\phi}+\boldsymbol{\nu}}{\Vert\boldsymbol{\phi}+\boldsymbol{\nu}\Vert}$, which ensures that the updated point stays on the manifold:
\begin{equation}
	\boldsymbol{\phi}^{(t+1)}=\mathcal{R}_{\boldsymbol{\phi}^{(t)}}\left(\tau \operatorname{grad}_{\mathcal{M}} G\left(\boldsymbol{\phi}^{(t)}\right)\right).
\end{equation}
This iterative procedure converges to a sub-optimal RIS phase shifts $\boldsymbol{\phi}^{\star}$.

\paragraph*{Computational complexity}
The overall complexity of the AO algorithm consists of two parts: the closed-form computation of the covariance matrix $\boldsymbol{R}$, and the RIS phase optimization using manifold gradient ascent. The computational complexity of constructing the covariance matrix $\boldsymbol{R}$ is dominated by the computation of $\boldsymbol{f}_1$ and $\boldsymbol{f}_2$, which requires $\mathcal{O}(MN)$ operations, and the construction of the outer products, which requires $\mathcal{O}(M^2)$ operations. Hence, the covariance matrix update costs $\mathcal{O}(MN + M^2)$. The computational complexity of the manifold-based algorithm for optimizing the RIS phase shifts is dominated by the computation of the Euclidean gradient, which requires $\mathcal{O}(MN)$ operations per iteration. The projection onto the tangent space and the retraction step each require $\mathcal{O}(N)$ operations. Thus, the manifold algorithm requires $\mathcal{O}(T_{\text{inner}} \cdot MN)$ operations per outer iteration, where $T_{\text{inner}}$ is the number of inner iterations. With $T_{\text{outer}}$ outer iterations, the total complexity is $\mathcal{O}\left(T_{\text{outer}} \left( T_{\text{inner}} \cdot MN + MN + M^2 \right) \right) = \mathcal{O}\left(T_{\text{outer}} T_{\text{inner}} \cdot MN\right)$, where the term $M^2$ is negligible when $N \gg M$.

\subsection{Asymptotic analysis}
\begin{assumption}[Quantified near-field distance regime]
	There exist constants $0<\rho_L\le \rho_U<\infty$ and $\psi_{\min}>0$, independent of $N$, such that
	\begin{equation}
		\rho_L d\sqrt{N} \le r_1 \le \rho_U d\sqrt{N},
	\end{equation}
	and
	\begin{equation}
		|\Psi_1| = |\sin\vartheta_1\cos\varphi_1| \ge \psi_{\min},
		\Phi_1 = \sin\vartheta_1\sin\varphi_1,
		\Omega_1 = \cos\vartheta_1.
	\end{equation}
\end{assumption}

\begin{remark}
	Assumption~1 is a quantified version of the usual statement $r_1=\Theta(d\sqrt{N})$, supplemented by a nondegeneracy condition on the directional factor $\Psi_1$. Its advantage is that the constants appearing in the lower-bound argument, especially those induced by the stationary-phase approximation\cite{stein1993harmonic}, can be made explicitly independent of $N$. Moreover, since both $r_1$ and the RIS aperture diameter scale as $\Theta(d\sqrt{N})$, the variation of the near-field channel amplitudes across RIS elements remains uniformly bounded. The three quantities $(\Phi_1,\Omega_1,\Psi_1)$ are the direction-cosine components of user~1 and satisfy
	\begin{equation}
		\Phi_1^2 + \Omega_1^2 + \Psi_1^2 = 1.
	\end{equation}
	Note that $|\Psi_1|\ge \psi_{\min}>0$ is an additional geometric nondegeneracy assumption rather than an automatic consequence of $\vartheta_1\in(0,\pi)$ and $\varphi_1\in(0,\pi)$. Its role is to exclude the degenerate stationary-phase case $\Psi_1=0$; only the magnitude $|\Psi_1|$ matters in the asymptotic scaling argument.
\end{remark}

\begin{assumption}[Effective RIS reflection geometry]
	The stationary point associated with the dominant LoS cascaded path from the BS to the RIS and then to the far-field user lies inside the RIS aperture.
\end{assumption}

\begin{remark}
	Assumption~2 ensures that the dominant reflected LoS component for the far-field user is effectively captured by the RIS. If it is violated, the corresponding reflected contribution becomes negligible, and the RIS no longer provides an effective specular path in this scenario.
\end{remark}

Consistent with the channel model in Section~II, we adopt the following specialization for asymptotic analysis: the BS-user links have only NLoS paths, i.e., $L_{k,\mathrm{BU}}=3$ for $k\in\{1,2\}$. The BS-RIS link has one LoS path plus three NLoS paths, i.e., $L_{\mathrm{BR}}=4$. The RIS-near-field-user channel is a single spherical-wave path, and the directional factor in its amplitude is understood through $|\Psi_k|$. The RIS-far-field-user channel is a single LoS path, i.e., $L_{2,\mathrm{RU}}=1$. The asymptotic proof only requires that the far-field links contain finitely many bounded-gain paths and that the near-field channel follows the spherical-wave model, so the proof remains valid under this specialization.

\begin{proposition}\label{prop:main}
	Suppose that the BS-RIS link and all direct links contain finitely many propagation paths, and that all associated path gains remain bounded independently of $M$ and $N$.
	
	Under Assumption~1, there exists a constant $\kappa_U>0$, independent of $M$, $N$, and $P_{\max}$, such that
	\begin{equation}
		C \le \log_2\!\left(1+\kappa_U P_{\max}MN\right).
	\end{equation}
	Furthermore, under Assumptions~1--2 and under the standard large-aperture stationary-phase approximation used in the lower-bound argument for the dominant LoS oscillatory sum, there exist a constant $\kappa_L>0$ and an integer $N_0$, both independent of $M$, $N$, and $P_{\max}$, such that for every $M\ge 1$ and all $N\ge N_0$,
	\begin{equation}
		C \ge \log_2\!\left(1+\kappa_L P_{\max}MN\right).
	\end{equation}
	Consequently,
	\begin{equation}
		C = \Theta\!\bigl(\log_2(MN)\bigr).
	\end{equation}
\end{proposition}

\begin{remark}
	The upper bound in Proposition~\ref{prop:main} is fully rigorous. The matching lower bound uses the standard large-aperture stationary-phase approximation for the dominant RIS-side oscillatory sum generated by the near-field compensation phase profile. The appendix proof is written directly in the covariance-matrix form of Problem $\mathcal P_2$, so it does not require assuming that the optimal covariance matrix is rank-one.
\end{remark}
\noindent\textit{Proof:} See Appendix~B.

\section{Simulation results}
The simulation parameters are set as follows. The carrier frequency is 28 GHz with a bandwidth of 100 MHz, and the noise power is $-94$ dBm. The BS is equipped with $8 \times 8$ antennas, and the RIS consists of $25 \times 25$ reflecting elements. The precision threshold for convergence is set to $10^{-4}$. All simulation results are averaged over $100$ independent channel realizations.

For the channel model, the direct BS–user link contains three NLoS paths, while the cascaded BS–RIS link comprises one LoS path and three NLoS paths. The complex gain of each path follows $\mathcal{CN}(0, 10^{-0.1 \cdot PL(dist)})$, where $PL(dist)$ is the distance-dependent path loss. The path loss is modeled as $PL(dist) = a + 10b \log_{10}(dist) + \xi$, with $\xi \sim \mathcal{N}(0, \sigma_{\xi}^2)$ representing shadow fading. For LoS paths, we set $a = 61.4$, $b = 2$, and $\sigma_{\xi} = 5.8$ dB; for NLoS paths, the parameters are $a = 72$, $b = 2.92$, and $\sigma_{\xi} = 8.7$ dB \cite{ref4}.

The BS is located at $(0, 0, 180)$ m, and the RIS is placed at the origin $(0, 0, 0)$. Both users are positioned along the same fixed direction, which is given by $\vartheta = \pi/3, \varphi = 2\pi/3$. The near-field user distance scales as $r_{\text{near}} = \rho d \sqrt{N}$ with $\rho = 27$, where $d = \lambda/2$. This ensures $r_{\text{near}} = \Theta(d\sqrt{N})$ as required by Assumption 1. The far-field user is placed at a fixed distance $r_{\text{far}} = 180$ m, which exceeds the maximum Rayleigh distance of the RIS which is approximately $167$ m for $N = 125^2$, guaranteeing that it always operates in the far-field region.

In addition, to verify the effectiveness of the proposed scheme and algorithm, we adopts the following comparative schemes and algorithms:
\begin{enumerate}
	\item NoRIS: the RIS is absent.
	\item FF-FF: All users are in the far-field of the RIS.
	\item NF-NF: All users are in the near-field of the RIS.
	\item NF-FF: In the proposed hybrid-field scenario, near-field and far-field users coexist.
	\item Fit: Linear fitting of the proposed scheme to extract the scaling slope $\log_2(MN)$.
	\item SCA-based AO: Alternating optimization for $\mathcal{P}_1$ with RIS phase shifts optimized via successive convex approximation (SCA)\cite{sca} and the convex subproblem corresponding to $\boldsymbol{R}$ solved by CVX\cite{cvx}.
\end{enumerate}

\begin{figure}[h]
	\centering
	\includegraphics[width=0.40\textwidth]{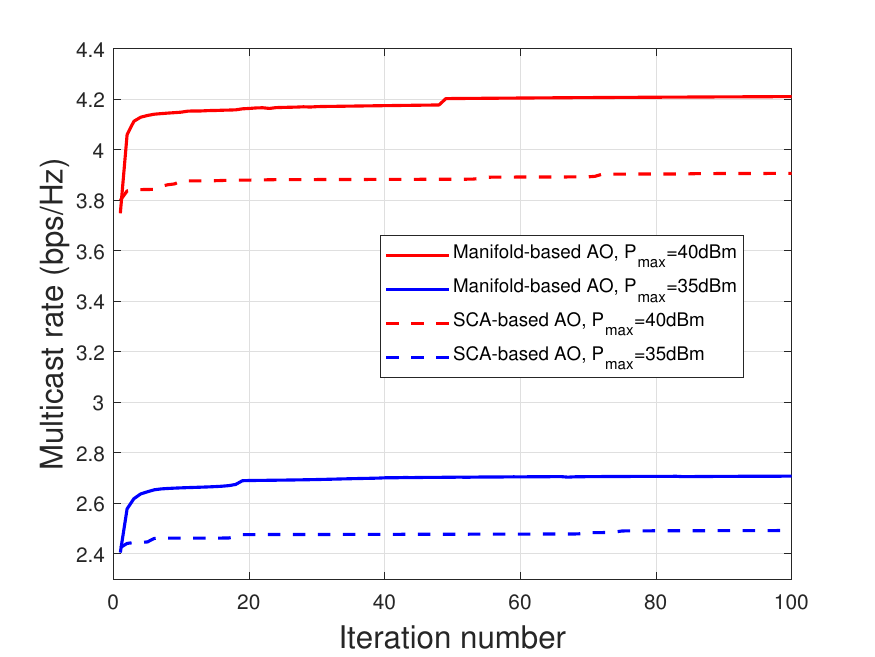}
	\caption{Convergence performance.} 
	\label{fig2}
\end{figure}
Fig. \ref{fig2} illustrates the convergence behavior of the proposed manifold-based AO algorithm compared with the SCA-based AO method. The proposed algorithm demonstrates rapid initial convergence and exhibits a monotonic non-decreasing trend. Notably, step-like increments appear in later iterations, such as around the $49$-th iteration, which validates the effectiveness of the alternating optimization strategy in further refining the solution. Throughout the process, the proposed manifold-based AO consistently outperforms the SCA-based approach. Specifically, at a transmit power of 40 dBm, it achieves a stable multicast rate of approximately 4.2 bps/Hz, significantly higher than the benchmark. Furthermore, increasing the transmit power from 35 dBm to 40 dBm yields a substantial performance gain, demonstrating the algorithm's ability to exploit power resources for capacity enhancement.

\begin{figure}[h]
	\centering
	\includegraphics[width=0.40\textwidth]{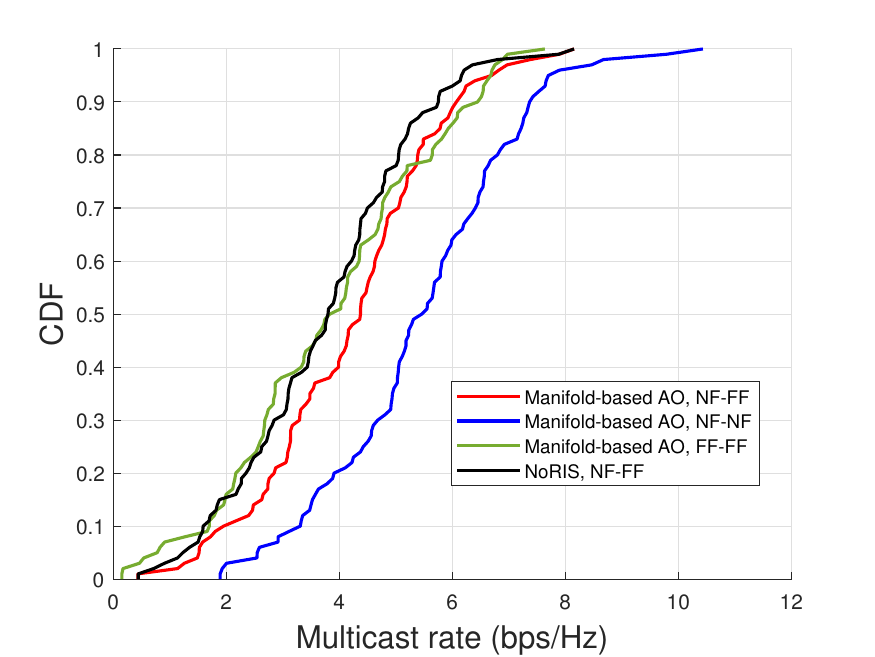}
	\caption{The CDF of multicast rate.} 
	\label{fig3}
\end{figure}
Fig. \ref{fig3} compares the cumulative distribution function (CDF) of the multicast rate for four scenarios. The NF-NF case achieves the highest rates across the entire range, benefiting from full RIS focusing gain. The NF-FF case exhibits a higher median rate than the FF-FF case, thanks to the RIS focusing that assists the near-field user. However, for the CDF values between 0.79 and 0.95, FF-FF overtakes NF-FF. This crossover occurs because the near–far field compromise caps the peak rate of NF-FF, while FF-FF can fully exploit the $N^2$ array gain \cite{intro11} in favorable channel realizations. All RIS-assisted scenarios significantly outperform the NoRIS baseline, confirming the necessity of RIS deployment.

\begin{figure}[h]
	\centering
	\includegraphics[width=0.40\textwidth]{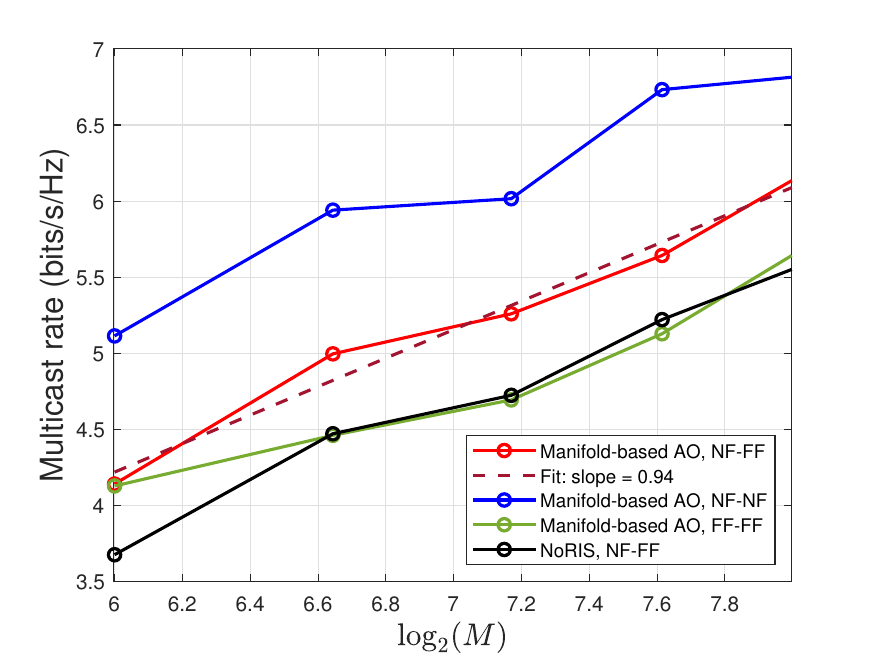}
	\caption{Multicast rate versus transmit antenna number.} 
	\label{fig4}
\end{figure}
Fig. \ref{fig4} illustrates the multicast rate versus the number of transmit antennas $M$ on a logarithmic scale. It is observed that the multicast rate increases with $M$ for all RIS-assisted schemes, demonstrating the effectiveness of increasing antennas to enhance system capacity. Specifically, the manifold-based AO with the NF-NF configuration achieves the highest rate, significantly outperforming all other schemes. The proposed manifold-based AO with the NF-FF configuration shows moderate performance, is superior to the FF-FF and NoRIS benchmarks. Moreover, the linear fit for the NF-FF scheme yields a slope of 0.94, which is close to the theoretical value of 1, validating the scalability of the proposed algorithm and confirming the $\Theta(\log_2(MN))$ scaling law even under the hybrid-field condition.

\begin{figure}[h]
	\centering
	\includegraphics[width=0.40\textwidth]{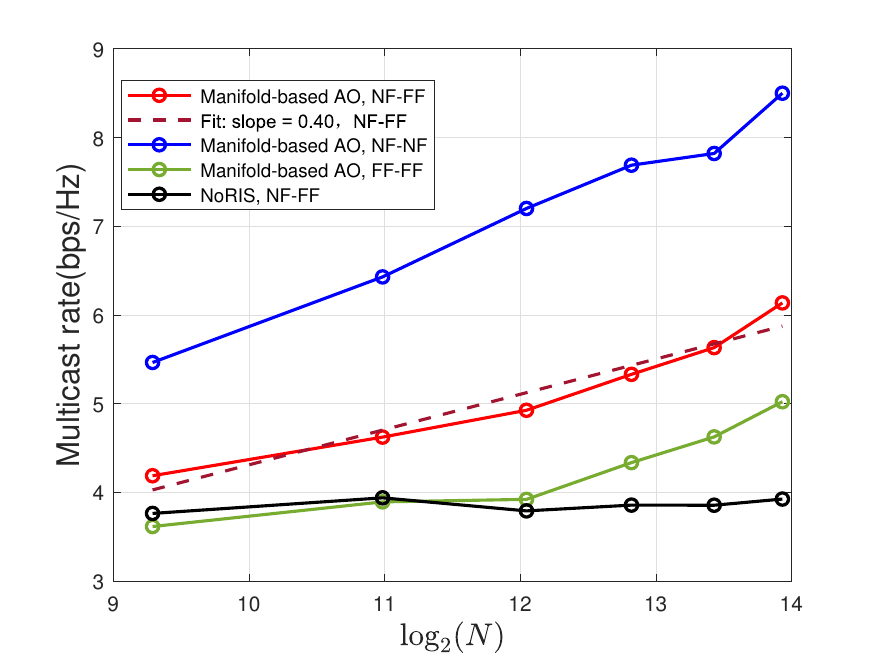}
	\caption{Multicast rate versus RIS element number.} 
	\label{fig5}
\end{figure}
Fig.~\ref{fig5} illustrates the multicast capacity as a function of $\log_2(N)$ with fixed $M$. Theoretically, for fixed $M$, the capacity satisfies $C = \Theta(\log_2(N))$, i.e., the asymptotic slope is $1$. Linear fitting of the three RIS-assisted curves yields slopes of $0.63$ for NF-NF, $0.40$ for NF-FF, and $0.28$ for FF-FF, respectively. All three slopes fall below the theoretical value of $1$, primarily because the asymptotic order analysis in this chapter pertains to the large-scale asymptotic regime, whereas the simulations over a finite range of $N$ remain in the pre-asymptotic regime. The decrease in slope from $0.63$ (NF-NF) to $0.40$ (NF-FF) reflects the additional gain loss induced by the RIS phase trade-off in hybrid-field scenarios, where the RIS must simultaneously serve both near-field and far-field users and thus cannot perfectly align with either user. The further decline in slope to $0.28$ (FF-FF) reveals the gain limitation of pure far-field beamforming under finite simulation conditions. Despite the slopes being lower than the theoretical value, the multicast capacity in the NF-FF scenario increases monotonically with $N$ and consistently exceeds the no-RIS baseline, thereby validating the practical effectiveness of the hybrid-field design.

\begin{figure}[h]
	\centering
	\includegraphics[width=0.40\textwidth]{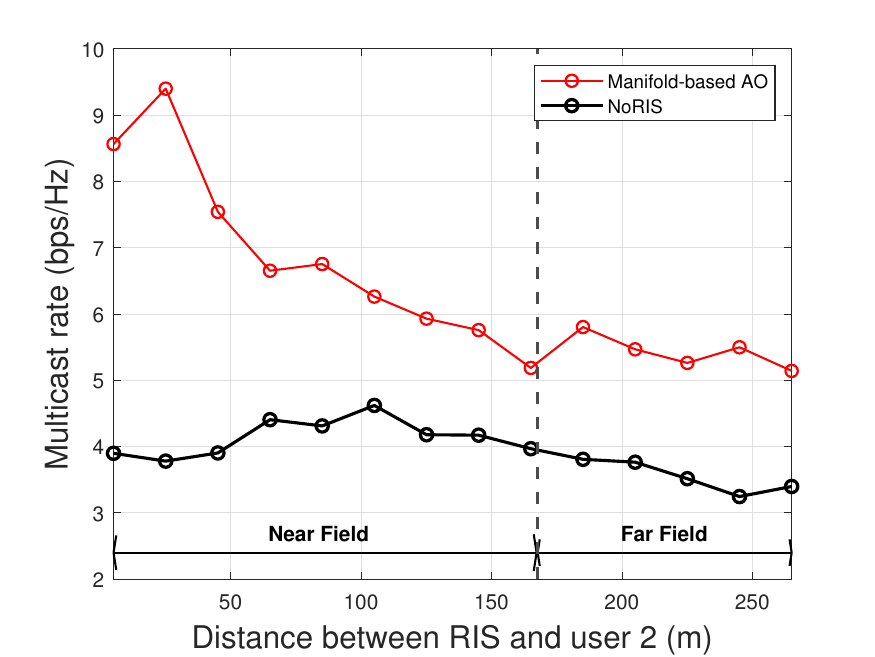}
	\caption{Multicast rate versus the distance between RIS and user $2$.} 
	\label{fig6}
\end{figure}
Fig.~\ref{fig6} illustrates the multicast capacity as a function of the distance between user~2 and the RIS. In this simulation, user~1 is fixed in the near-field region of the RIS, while the distance of user~2 increases from 5~m to 265~m, causing user~2 to transition gradually from the near-field to the far-field region of the RIS. When user~2 is located in the near-field, the RIS can achieve energy focusing for this user, resulting in a high multicast rate that reaches approximately 9.4~bps/Hz at 25~m. As user~2 enters the far-field region, the focusing gain gradually diminishes, and the multicast rate decreases to a stable level of approximately 5.1~bps/Hz in the far-field, yet remains significantly higher than the no-RIS baseline. These results validate the effectiveness of the proposed hybrid-field design, which is reflected in the following three aspects. First, as user~2 moves from the near-field to the far-field, the multicast rate remains substantially higher than the NoRIS baseline, indicating that the RIS can simultaneously serve both near-field and far-field users with a single phase profile. Second, the multicast rate is significantly higher when user~2 is in the near-field, demonstrating the practical gain offered by near-field focusing. Third, the multicast rate decreases and stabilizes as user~2 enters the far-field, suggesting that system performance in multicast scenarios is primarily limited by the far-field user, which aligns with the fundamental characteristic that the multicast rate is constrained by the weakest user. Collectively, these results demonstrate that the proposed hybrid-field design can effectively leverage both near-field focusing and far-field beamforming capabilities of the RIS to achieve reliable multicast transmission in hybrid-field scenarios.

\section{Conclusion}
The paper establishes the fundamental multicast capacity limit for XL-RIS-assisted mmWave communications in hybrid near-field and far-field scenarios. For the canonical two-user case, we derive a closed-form expression for the optimal covariance matrix and propose a manifold optimization algorithm for RIS phase shift optimization. We prove that the asymptotic order of the multicast capacity is $\Theta(\log_2(MN))$, and that this order is tight. Numerical results verify the tightness of the derived bounds and quantify the impact of $M$, $N$, and distance on the achievable multicast rate. 

\begin{appendices}
\section{PROOF of PROPOSITION \ref{prop:kkt}}
\begin{proof}
	If $\lambda_1=0$, then $\lambda_2=1$ based on \eqref{eq15}. Thus, \eqref{eq14} can be rewritten as
	\begin{equation}\label{eq20}
		\lambda_2\boldsymbol{f}_2\boldsymbol{f}_2^H=\mu\boldsymbol{I}-\boldsymbol{Z}.
	\end{equation} 
	The left side of \eqref{eq20} is rank-one constraint because $\boldsymbol{f}_1$ is not zero. If $\mu>0$, $\mu\boldsymbol{I}$ is full-rank. Assuming the eigenvalues of $\boldsymbol{Z}$ is $v_1\geq v_2\geq\cdots\geq v_M \geq 0$, then the eigenvalues of  $\mu\boldsymbol{I}-\boldsymbol{Z}$ is $\mu-v_1,\mu-v_2,\cdots,\mu-v_M$. Because the rank of $\mu\boldsymbol{I}-\boldsymbol{Z}$ is $1$, thus $\mu-v_1\neq 0, \mu-v_2=0,\cdots,\mu-v_M=0$. Accordingly, we have $v_1\geq v_2= \cdots= v_M = \mu$. Moreover, the left side $\lambda_2\boldsymbol{f}_2\boldsymbol{f}_2^H$ is positive semi-definite, thus its nonzero eigenvalue is positive, i.e., $\mu-v_1\geq 0$. Combining $v_1\geq \mu$, we have $v_1=\mu$. It contradicts $\mu\neq v_1$, therefore the assumption $\lambda_1=0, \lambda_2=1$ is invalid. If $\mu=0$, then the right side of \eqref{eq20} becomes $-\boldsymbol{Z}$, which is negative semi-definite. It is not equivalent to the conclusion that it is positive semi-definite. Thus, the assumption $\lambda_1=0, \lambda_2=1$ is invalid in both cases. Similarly, the assumption $\lambda_1=1, \lambda_2=0$ does not hold, either. In conclusion, $\lambda_1> 0, \lambda_2> 0$.
\end{proof}

\section{PROOF of PROPOSITION \ref{prop:main}}
\label{app:B}

In this appendix, we prove the asymptotic proposition. The proof consists of two parts. We first establish a fully rigorous upper bound on the multicast capacity. We then construct a feasible RIS phase profile and a feasible rank-one covariance matrix to derive a matching lower bound. The lower-bound part uses the standard large-aperture stationary-phase approximation for the dominant RIS-side oscillatory sum.

\subsection{Setup and auxiliary quantities}
For a fixed RIS phase vector $\boldsymbol{\phi}$, define the effective channel vectors
\begin{equation}
	\boldsymbol f_k(\boldsymbol{\phi}) = \frac{\boldsymbol H^H\boldsymbol \Phi^H\boldsymbol h_k + \boldsymbol d_k}{\sigma_k},
	\qquad k=1,2,
\end{equation}
where $\boldsymbol \Phi = \diag(\boldsymbol{\phi})$, and let
\begin{equation}
	a(\boldsymbol{\phi}) = \lVert \boldsymbol f_1(\boldsymbol{\phi})\rVert^2,
	\qquad
	b(\boldsymbol{\phi}) = \lVert \boldsymbol f_2(\boldsymbol{\phi})\rVert^2.
\end{equation}
For fixed $\boldsymbol{\phi}$, define the optimal common receive gain directly in the covariance-matrix form as
\begin{equation}
	\Gamma^{\star}(\boldsymbol{\phi}) := \max_{\boldsymbol R\succeq 0,\, \Tr(\boldsymbol R)\le P_{\max}}
	\min\left\{\boldsymbol f_1^H(\boldsymbol{\phi})\boldsymbol R\boldsymbol f_1(\boldsymbol{\phi}),\,
	\boldsymbol f_2^H(\boldsymbol{\phi})\boldsymbol R\boldsymbol f_2(\boldsymbol{\phi})\right\}.
\end{equation}
Hence the multicast capacity is
\begin{equation}
	C = \max_{\boldsymbol{\phi}} \log_2\!\bigl(1+\Gamma^{\star}(\boldsymbol{\phi})\bigr).
\end{equation}

We next collect several bounded quantities that will be used repeatedly.

\paragraph*{1) BS-RIS spectral norm bound}
The BS-RIS channel can be written as
\begin{equation}
	\boldsymbol H = \sqrt{\frac{MN}{L_{\mathrm{BR}}}}\sum_{\ell=0}^{L_{\mathrm{BR}}-1}
	\beta_{\ell}\, \boldsymbol a_{R,\ell}\boldsymbol a_{B,\ell}^H,
\end{equation}
where $L_{\mathrm{BR}}$ is fixed and the path gains $\{\beta_{\ell}\}$ are bounded. Define
\begin{equation}
	\bar H := \frac{1}{\sqrt{L_{\mathrm{BR}}}}\sum_{\ell=0}^{L_{\mathrm{BR}}-1}|\beta_{\ell}|.
\end{equation}
Since the steering vectors have unit norm,
\begin{equation}
	\lVert \boldsymbol H\rVert_2 \le \sqrt{\frac{MN}{L_{\mathrm{BR}}}}
	\sum_{\ell=0}^{L_{\mathrm{BR}}-1}|\beta_{\ell}| = \bar H\sqrt{MN}.
\end{equation}
In the simulation-consistent specialization adopted above, $L_{\mathrm{BR}}=4$, with one LoS path and three NLoS paths, but the proof only needs $L_{\mathrm{BR}}<\infty$ and bounded path gains.

\paragraph*{2) Direct-link bounds}
For each user $k\in\{1,2\}$, the direct BS-user link is far-field with a finite number of bounded-gain propagation paths. Hence there exists a constant $\bar d_k>0$, independent of $M$ and $N$, such that
\begin{equation}
	\lVert \boldsymbol d_k\rVert \le \bar d_k\sqrt{M},
	\qquad k=1,2.
\end{equation}
In the simulation-consistent specialization, each direct link contains three NLoS paths and no LoS path, i.e., $L_{k,\mathrm{BU}}=3$, and
\begin{equation}
	\boldsymbol d_k = \sqrt{\frac{M}{3}}\sum_{\ell=1}^{3}\alpha_{\ell,k}\boldsymbol a_{B,\ell,k},
	\qquad k=1,2,
\end{equation}
so that one may take
\begin{equation}
	\bar d_k = \frac{1}{\sqrt 3}\sum_{\ell=1}^{3}|\alpha_{\ell,k}|.
\end{equation}

\paragraph*{3) Near-field user norm and amplitude-ratio bounds}
For the RIS-near-field user, the channel is modeled as a single-path spherical-wave channel, and the channel amplitude satisfies
\begin{equation}
	|[\boldsymbol h_1]_n| \propto r_{n,1}^{-3/2}.
\end{equation}
Under Assumption~1, the near-field channel power remains bounded away from both zero and infinity in the large-aperture regime. Therefore, there exist constants $\underline h_1>0$ and $\bar h_1>0$, independent of $N$, such that
\begin{equation}
	\underline h_1 \le \lVert \boldsymbol h_1\rVert \le \bar h_1,
	\qquad \forall N.
\end{equation}
Let $D_{\mathrm{RIS}}$ denote the RIS aperture diameter. Since $D_{\mathrm{RIS}}=\Theta(d\sqrt N)$, there exists a constant $c_D>0$, independent of $N$, such that
\begin{equation}
	D_{\mathrm{RIS}} \le c_D d\sqrt N.
\end{equation}
By simple geometry, every RIS element satisfies
\begin{equation}
	r_1 \le r_{n,1} \le \sqrt{r_1^2 + D_{\mathrm{RIS}}^2}.
\end{equation}
Using Assumption~1, we obtain
\begin{equation}
	\frac{\max_n r_{n,1}}{\min_n r_{n,1}} \le
	\sqrt{1+\frac{D_{\mathrm{RIS}}^2}{r_1^2}}
	\le \sqrt{1+\frac{c_D^2}{\rho_L^2}}.
\end{equation}
Define the amplitude ratio
\begin{equation}
	\eta_N := \frac{\max_n |[\boldsymbol h_1]_n|}{\min_n |[\boldsymbol h_1]_n|}.
\end{equation}
Since $|[\boldsymbol h_1]_n|\propto r_{n,1}^{-3/2}$,
\begin{equation}
	\eta_N \le
	\left(\frac{\max_n r_{n,1}}{\min_n r_{n,1}}\right)^{3/2}
	\le \left(1+\frac{c_D^2}{\rho_L^2}\right)^{3/4}.
\end{equation}
Hence there exists a constant $\bar\eta>0$ such that
\begin{equation}
	\eta_N \le \bar\eta,
	\qquad \forall N.
\end{equation}

\subsection{Upper bound}
We first derive the rigorous upper bound.

Fix an arbitrary RIS phase vector $\boldsymbol{\phi}$ and an arbitrary feasible covariance matrix $\boldsymbol R\succeq 0$ with $\Tr(\boldsymbol R)\le P_{\max}$. Then
\begin{equation}
	\min\left\{\boldsymbol f_1^H\boldsymbol R\boldsymbol f_1,\,\boldsymbol f_2^H\boldsymbol R\boldsymbol f_2\right\}
	\le \boldsymbol f_1^H\boldsymbol R\boldsymbol f_1.
\end{equation}
Since $\boldsymbol R\succeq 0$,
\begin{equation}
	\boldsymbol f_1^H\boldsymbol R\boldsymbol f_1 \le \lambda_{\max}(\boldsymbol R)\,\lVert \boldsymbol f_1\rVert^2.
\end{equation}
Moreover,
\begin{equation}
	\lambda_{\max}(\boldsymbol R) \le \Tr(\boldsymbol R) \le P_{\max}.
\end{equation}
Hence
\begin{equation}
	\min\left\{\boldsymbol f_1^H\boldsymbol R\boldsymbol f_1,\,\boldsymbol f_2^H\boldsymbol R\boldsymbol f_2\right\}
	\le P_{\max} a(\boldsymbol{\phi}).
\end{equation}
Since this holds for every feasible $\boldsymbol R$, maximizing over $\boldsymbol R$ gives
\begin{equation}
	\Gamma^{\star}(\boldsymbol{\phi}) \le P_{\max} a(\boldsymbol{\phi}).
\end{equation}
We next upper-bound $a(\boldsymbol{\phi})$. From the definition,
\begin{equation}
	\sqrt{a(\boldsymbol{\phi})} = \frac{\lVert \boldsymbol H^H\boldsymbol \Phi^H\boldsymbol h_1 + \boldsymbol d_1\rVert}{\sigma_1}.
\end{equation}
Using the triangle inequality,
\begin{equation}
	\lVert \boldsymbol H^H\boldsymbol \Phi^H\boldsymbol h_1 + \boldsymbol d_1\rVert
	\le \lVert \boldsymbol H^H\boldsymbol \Phi^H\boldsymbol h_1\rVert + \lVert \boldsymbol d_1\rVert.
\end{equation}
The first term satisfies
\begin{equation}
	\lVert \boldsymbol H^H\boldsymbol \Phi^H\boldsymbol h_1\rVert
	\le \lVert \boldsymbol H^H\rVert_2\lVert \boldsymbol \Phi^H\rVert_2\lVert \boldsymbol h_1\rVert
	= \lVert \boldsymbol H\rVert_2\lVert \boldsymbol h_1\rVert
	\le \bar H\bar h_1\sqrt{MN},
\end{equation}
where we used $\lVert \boldsymbol \Phi\rVert_2=1$. For the direct term,
\begin{equation}
	\lVert \boldsymbol d_1\rVert \le \bar d_1\sqrt M.
\end{equation}
Since $N\ge 1$, we have $\sqrt M\le \sqrt{MN}$. Therefore,
\begin{equation}
	\lVert \boldsymbol H^H\boldsymbol \Phi^H\boldsymbol h_1 + \boldsymbol d_1\rVert
	\le (\bar H\bar h_1 + \bar d_1)\sqrt{MN}.
\end{equation}
Squaring both sides and dividing by $\sigma_1^2$ yields
\begin{equation}
	a(\boldsymbol{\phi}) \le \frac{(\bar H\bar h_1 + \bar d_1)^2}{\sigma_1^2}MN
	=: \bar A_1 MN.
\end{equation}
Combining the previous displays, we obtain
\begin{equation}
	\Gamma^{\star}(\boldsymbol{\phi}) \le \bar A_1 P_{\max}MN.
\end{equation}
Hence,
\begin{equation}
	C \le \log_2\!\bigl(1+\bar A_1 P_{\max}MN\bigr).
\end{equation}
Therefore, the upper bound in the asymptotic proposition holds with $\kappa_U=\bar A_1$.

\subsection{Lower bound}
We next derive the matching lower bound.

\paragraph*{1) Construction of the RIS phase profile}
Let $\ell=0$ denote the dominant LoS path of the BS-RIS channel. We choose
\begin{equation}
	[\boldsymbol{\phi}_0]_n = \exp\!\left(j\frac{2\pi}{\lambda}r_{n,1}\right)
	\frac{[\boldsymbol a_{R,0}]_n^{\ast}}{|[\boldsymbol a_{R,0}]_n|}.
\end{equation}
This choice performs two tasks simultaneously: it compensates the spherical-wave phase of the RIS-near-field user and aligns the dominant LoS BS-RIS steering vector. Define
\begin{equation}
	a_0 := a(\boldsymbol{\phi}_0),
	\qquad
	b_0 := b(\boldsymbol{\phi}_0).
\end{equation}
We will show that both $a_0$ and $b_0$ are at least of order $MN$.

\paragraph*{2) Lower bound on $a_0$}
Define
\begin{equation}
	x_0(\boldsymbol{\phi}_0) := \frac{1}{\sqrt N}\sum_{n=1}^{N}|[\boldsymbol h_1]_n|.
\end{equation}
Because the phase in the chosen $\boldsymbol{\phi}_0$ cancels the near-field propagation phase and aligns the dominant RIS steering vector, the LoS contribution to the cascaded near-field channel adds coherently. We now lower-bound $x_0(\boldsymbol{\phi}_0)$. By definition,
\begin{equation}
	\sum_{n=1}^{N}|[\boldsymbol h_1]_n| \ge N\min_n |[\boldsymbol h_1]_n|.
\end{equation}
Also,
\begin{equation}
	\max_n |[\boldsymbol h_1]_n| \ge \frac{\lVert \boldsymbol h_1\rVert}{\sqrt N}.
\end{equation}
Using the amplitude-ratio bound,
\begin{equation}
	\min_n |[\boldsymbol h_1]_n| \ge \frac{1}{\bar\eta}\max_n |[\boldsymbol h_1]_n|
	\ge \frac{\lVert \boldsymbol h_1\rVert}{\bar\eta\sqrt N}.
\end{equation}
Therefore,
\begin{equation}
	\sum_{n=1}^{N}|[\boldsymbol h_1]_n| \ge N\cdot \frac{\lVert \boldsymbol h_1\rVert}{\bar\eta\sqrt N}
	= \frac{\lVert \boldsymbol h_1\rVert}{\bar\eta}\sqrt N
	\ge \frac{\underline h_1}{\bar\eta}\sqrt N.
\end{equation}
Dividing by $\sqrt N$ yields
\begin{equation}
	x_0(\boldsymbol{\phi}_0) \ge \frac{\underline h_1}{\bar\eta}
	=: x_0^{\circ}>0.
\end{equation}
Now consider the cascaded term $\boldsymbol H^H\boldsymbol \Phi_0^H\boldsymbol h_1$. Retaining only the dominant LoS BS-RIS path in the expansion of $\boldsymbol H$, the coherent contribution gives
\begin{equation}
	\lVert \boldsymbol H^H\boldsymbol \Phi_0^H\boldsymbol h_1\rVert^2
	\ge \frac{MN}{L_{\mathrm{BR}}}|\beta_0|^2|x_0(\boldsymbol{\phi}_0)|^2.
\end{equation}
Hence
\begin{equation}
	\lVert \boldsymbol H^H\boldsymbol \Phi_0^H\boldsymbol h_1\rVert
	\ge \frac{|\beta_0|x_0^{\circ}}{\sqrt{L_{\mathrm{BR}}}}\sqrt{MN}
	=: q_1\sqrt{MN}.
\end{equation}
We next reintroduce the direct link through the reverse triangle inequality:
\begin{equation}
	\lVert \boldsymbol H^H\boldsymbol \Phi_0^H\boldsymbol h_1+\boldsymbol d_1\rVert
	\ge \lVert \boldsymbol H^H\boldsymbol \Phi_0^H\boldsymbol h_1\rVert-\lVert \boldsymbol d_1\rVert
	\ge q_1\sqrt{MN}-\bar d_1\sqrt M.
\end{equation}
If
\begin{equation}
	N \ge N_1 := \left(\frac{2\bar d_1}{q_1}\right)^2,
\end{equation}
then
\begin{equation}
	q_1\sqrt{MN} - \bar d_1\sqrt M \ge \frac{q_1}{2}\sqrt{MN}.
\end{equation}
Therefore, for all $N\ge N_1$,
\begin{equation}
	a_0 = \frac{\lVert \boldsymbol H^H\boldsymbol \Phi_0^H\boldsymbol h_1 + \boldsymbol d_1\rVert^2}{\sigma_1^2}
	\ge \frac{q_1^2}{4\sigma_1^2}MN
	=: A_1MN.
\end{equation}

\paragraph*{3) Lower bound on $b_0$ under the single-LoS far-field model}
Under the updated specialization, the RIS-far-field-user channel is
\begin{equation}
	\boldsymbol h_2 = \sqrt N\, \beta_2\boldsymbol a_{R,2},
\end{equation}
where $\boldsymbol a_{R,2}$ is a unit-norm far-field steering vector and $\beta_2$ is bounded independently of $N$. Hence,
\begin{equation}
	\boldsymbol H^H\boldsymbol \Phi_0^H\boldsymbol h_2 =
	\sqrt{\frac{MN^2}{L_{\mathrm{BR}}}}\,\beta_2^{\ast}
	\sum_{\ell=0}^{L_{\mathrm{BR}}-1}\beta_{\ell}^{\ast}y_{\ell}^{\ast}(\boldsymbol{\phi}_0)\boldsymbol a_{B,\ell},
\end{equation}
where the RIS-side coupling coefficient is defined as
\begin{equation}
	y_{\ell}(\boldsymbol{\phi}_0) := \boldsymbol a_{R,2}^H\boldsymbol \Phi_0\boldsymbol a_{R,\ell}.
\end{equation}
For the dominant LoS term $\ell=0$, the quantity $Ny_0(\boldsymbol{\phi}_0)$ is a two-dimensional oscillatory lattice sum whose phase consists of a linear far-field part and a quadratic Fresnel term induced by the near-field compensation. The quadratic part can be written as
\begin{equation}
	f^{(2)}(n_y,n_z) = \frac{\pi d^2}{\lambda r_1}
	\left(n_y^2 c_y + n_z^2 c_z - 2n_yn_z\Phi_1\Omega_1\right),
\end{equation}
where
\begin{equation}
	c_y = 1-\Phi_1^2,
	\qquad
	c_z = 1-\Omega_1^2.
\end{equation}
The Hessian determinant of this quadratic form is
\begin{equation}
	\det(H_f) = \left(\frac{2\pi d^2}{\lambda r_1}\right)^2
	\left(c_yc_z-\Phi_1^2\Omega_1^2\right)
	= \left(\frac{2\pi d^2}{\lambda r_1}\right)^2\Psi_1^2>0,
\end{equation}
where we used
\begin{equation}
	c_yc_z-\Phi_1^2\Omega_1^2 = 1-\Phi_1^2-\Omega_1^2 = \Psi_1^2,
\end{equation}
and Assumption~1 ensures $|\Psi_1|\ge \psi_{\min}>0$, so the stationary point is nondegenerate. Assumption~2 guarantees that the corresponding stationary point lies inside the RIS aperture. Since only $|\Psi_1|$ enters through $|\det(H_f)|^{-1/2}$, the stationary-phase constants depend on $|\Psi_1|$ rather than on the sign of $\Psi_1$.

By the standard large-aperture stationary-phase approximation, there exist constants $C_y^{(\ell)}>0$, $C_y^{(u)}>0$, and an integer $N_y$, all independent of $N$, such that for every $N\ge N_y$,
\begin{equation}
	C_y^{(\ell)}\sqrt N \le |y_0(\boldsymbol{\phi}_0)| \le C_y^{(u)}\sqrt N.
\end{equation}
Retaining only the dominant LoS term $\ell=0$ in the previous expansion gives
\begin{equation}
	\lVert \boldsymbol H^H\boldsymbol \Phi_0^H\boldsymbol h_2\rVert^2
	\ge \frac{MN^2}{L_{\mathrm{BR}}}|\beta_2|^2|\beta_0|^2|y_0(\boldsymbol{\phi}_0)|^2.
\end{equation}
Using the stationary-phase lower bound, we obtain for all $N\ge N_y$,
\begin{equation}
	\lVert \boldsymbol H^H\boldsymbol \Phi_0^H\boldsymbol h_2\rVert^2
	\ge \frac{|\beta_2|^2|\beta_0|^2(C_y^{(\ell)})^2}{L_{\mathrm{BR}}}MN.
\end{equation}
Define
\begin{equation}
	q_2 := \frac{|\beta_2||\beta_0|C_y^{(\ell)}}{\sqrt{L_{\mathrm{BR}}}}.
\end{equation}
Then
\begin{equation}
	\lVert \boldsymbol H^H\boldsymbol \Phi_0^H\boldsymbol h_2\rVert \ge q_2\sqrt{MN},
	\qquad N\ge N_y.
\end{equation}
We again reintroduce the direct link:
\begin{equation}
	\lVert \boldsymbol H^H\boldsymbol \Phi_0^H\boldsymbol h_2+\boldsymbol d_2\rVert
	\ge \lVert \boldsymbol H^H\boldsymbol \Phi_0^H\boldsymbol h_2\rVert-\lVert \boldsymbol d_2\rVert
	\ge q_2\sqrt{MN}-\bar d_2\sqrt M.
\end{equation}
If
\begin{equation}
	N \ge N_2 := \left(\frac{2\bar d_2}{q_2}\right)^2,
\end{equation}
then
\begin{equation}
	q_2\sqrt{MN} - \bar d_2\sqrt M \ge \frac{q_2}{2}\sqrt{MN}.
\end{equation}
Therefore, for all $N\ge \max\{N_y,N_2\}$,
\begin{equation}
	b_0 = \frac{\lVert \boldsymbol H^H\boldsymbol \Phi_0^H\boldsymbol h_2+\boldsymbol d_2\rVert^2}{\sigma_2^2}
	\ge \frac{q_2^2}{4\sigma_2^2}MN
	=: A_2MN.
\end{equation}

\paragraph*{4) Feasible covariance construction and lower bound on $\Gamma^{\star}(\boldsymbol{\phi}_0)$}
At this point we know that
\begin{equation}
	\lVert \boldsymbol f_1(\boldsymbol{\phi}_0)\rVert^2 = a_0 \ge A_1MN,
	\qquad
	\lVert \boldsymbol f_2(\boldsymbol{\phi}_0)\rVert^2 = b_0 \ge A_2MN.
\end{equation}
Define the normalized vectors
\begin{equation}
	\boldsymbol u_1 := \frac{\boldsymbol f_1(\boldsymbol{\phi}_0)}{\lVert \boldsymbol f_1(\boldsymbol{\phi}_0)\rVert},
	\qquad
	\boldsymbol u_2 := e^{-j\arg(\boldsymbol f_1^H(\boldsymbol{\phi}_0)\boldsymbol f_2(\boldsymbol{\phi}_0))}
	\frac{\boldsymbol f_2(\boldsymbol{\phi}_0)}{\lVert \boldsymbol f_2(\boldsymbol{\phi}_0)\rVert},
\end{equation}
and choose
\begin{equation}
	\boldsymbol w_0 := \frac{\boldsymbol u_1+\boldsymbol u_2}{\lVert \boldsymbol u_1+\boldsymbol u_2\rVert}.
\end{equation}
Since $\lVert \boldsymbol u_1\rVert = \lVert \boldsymbol u_2\rVert = 1$ and the phase rotation makes $\boldsymbol u_1^H\boldsymbol u_2\in[0,1]$, we have
\begin{equation}
	|\boldsymbol f_1^H(\boldsymbol{\phi}_0)\boldsymbol w_0|^2 \ge \frac{a_0}{2},
	\qquad
	|\boldsymbol f_2^H(\boldsymbol{\phi}_0)\boldsymbol w_0|^2 \ge \frac{b_0}{2}.
\end{equation}
Now define the feasible rank-one covariance matrix
\begin{equation}
	\boldsymbol R_0 := P_{\max}\boldsymbol w_0\boldsymbol w_0^H.
\end{equation}
Then $\boldsymbol R_0\succeq 0$ and
\begin{equation}
	\Tr(\boldsymbol R_0)=P_{\max},
\end{equation}
so $\boldsymbol R_0$ is feasible. Therefore,
\begin{equation}
	\begin{aligned}
	\Gamma^{\star}(\boldsymbol{\phi}_0)&\ge \min\left\{\boldsymbol f_1^H(\boldsymbol{\phi}_0)\boldsymbol R_0\boldsymbol f_1(\boldsymbol{\phi}_0),\,
	\boldsymbol f_2^H(\boldsymbol{\phi}_0)\boldsymbol R_0\boldsymbol f_2(\boldsymbol{\phi}_0)\right\}\\
	&\ge \frac{P_{\max}}{2}\min\{a_0,b_0\}.
	\end{aligned}
\end{equation}
Combining the lower bounds on $a_0$ and $b_0$, we obtain for all
\begin{equation}
	N \ge N_0 := \max\{N_1,N_2,N_y\}
\end{equation}
that
\begin{equation}
	\Gamma^{\star}(\boldsymbol{\phi}_0) \ge \frac{P_{\max}}{2}\min\{A_1,A_2\}MN.
\end{equation}
Define
\begin{equation}
	\kappa_L := \frac{1}{2}\min\{A_1,A_2\}.
\end{equation}
Then
\begin{equation}
	\Gamma^{\star}(\boldsymbol{\phi}_0) \ge \kappa_L P_{\max}MN,
	\qquad N\ge N_0.
\end{equation}
Finally, since $C$ maximizes over all feasible RIS phase vectors,
\begin{equation}
	C \ge \log_2\!\bigl(1+\Gamma^{\star}(\boldsymbol{\phi}_0)\bigr)
	\ge \log_2\!\bigl(1+\kappa_L P_{\max}MN\bigr),
	\quad N\ge N_0.
\end{equation}

\subsection{Conclusion}
The rigorous upper bound and the constructive lower bound show that for every $M\ge 1$ and all $N\ge N_0$,
\begin{equation}
	\log_2\!\bigl(1+\kappa_L P_{\max}MN\bigr)
	\le C \le
	\log_2\!\bigl(1+\kappa_U P_{\max}MN\bigr),
\end{equation}
where $\kappa_L$, $\kappa_U$, and $N_0$ are independent of $M$, $N$, and $P_{\max}$. This proves the asymptotic proposition.

\begin{remark}
	An alternative lower-bound route starts from the exact fixed-$\boldsymbol{\phi}$ closed-form SINR and then requires a sufficiently sharp upper bound on the cross-correlation term $c_0 = \boldsymbol f_1^H(\boldsymbol{\phi}_0)\boldsymbol f_2(\boldsymbol{\phi}_0)$. That route is valid in principle but technically more delicate in the present hybrid near-/far-field setting. The covariance-construction route adopted here bypasses that step entirely.
\end{remark}
\end{appendices}

\section*{ACKNOWLEDGMENTS}
This work was supported by the National Key R\&D Program of China under grant 2021YFA0716600. (Corresponding author: Hongcheng Zhuang)


\begin{thebibliography}{1}
\bibliographystyle{IEEEtran}
\bibitem{intro1}
F. Qi, Q. Huang, Y. Wang, P. Lin, C. Hu and X. Shi, ``IRS-Assisted Joint Broadcast/Multicast Resource Allocation and Cooperative Perception Optimization for 6G V2X Networks," in IEEE Transactions on Broadcasting, vol. 72, no. 1, pp. 106-116, March 2026.

\bibitem{intro2}
N. Q. Hieu, D. N. Nguyen, D. T. Hoang and E. Dutkiewicz, ``When Virtual Reality Meets Rate Splitting Multiple Access: A Joint Communication and Computation Approach," in IEEE Journal on Selected Areas in Communications, vol. 41, no. 5, pp. 1536-1548, May 2023.

\bibitem{intro3}
N. Chukhno et al., ``Models, Methods, and Solutions for Multicasting in 5G/6G mmWave and Sub-THz Systems," in IEEE Communications Surveys \& Tutorials, vol. 26, no. 1, pp. 119-159, Firstquarter 2024.

\bibitem{intro4}
K. Zhi et al., ``Performance Analysis and Low-Complexity Design for XL-MIMO With Near-Field Spatial Non-Stationarities," in IEEE Journal on Selected Areas in Communications, vol. 42, no. 6, pp. 1656-1672, June 2024.

\bibitem{intro5}
Y. Liu, Z. Wang, J. Xu, C. Ouyang, X. Mu and R. Schober, ``Near-Field Communications: A Tutorial Review," in IEEE Open Journal of the Communications Society, vol. 4, pp. 1999-2049, 2023.



\bibitem{intro8}
R. Li, L. Zhao, M. Li, M. -M. Zhao and D. W. K. Ng, ``Sensing-Based Channel Estimation for Extremely Large-Scale RIS-Assisted Millimeter-Wave Communication Systems," in IEEE Internet of Things Journal, vol. 12, no. 23, pp. 50732-50746, 1 Dec.1, 2025.

\bibitem{intro8-1}
T. Van Chien, L. T. Tu, S. Chatzinotas and B. Ottersten, ``Coverage Probability and Ergodic Capacity of Intelligent Reflecting Surface-Enhanced Communication Systems," in IEEE Communications Letters, vol. 25, no. 1, pp. 69-73, Jan. 2021.

\bibitem{intro8-2}
S. Zhang, Z. Yang, M. Chen, D. Liu, K. -K. Wong and H. V. Poor, ``Beamforming Design for the Performance Optimization of Intelligent Reflecting Surface Assisted Multicast MIMO Networks," in IEEE Transactions on Wireless Communications, vol. 23, no. 3, pp. 2325-2339, March 2024.

\bibitem{intro8-3}
X. Yu, W. Shen, R. Zhang, C. Xing and T. Q. S. Quek, ``Channel Estimation for XL-RIS-Aided Millimeter-Wave Systems," in IEEE Transactions on Communications, vol. 71, no. 9, pp. 5519-5533, Sept. 2023.

\bibitem{intro8-4}
X. Bian, W. Xu, Y. Wang, C. Yuen and Z. Cai, ``A Sparse-Based Channel Estimation Scheme for XL-RIS-Assisted Wireless Communications With Low Pilot Overheads," in IEEE Transactions on Communications, vol. 74, pp. 367-380, 2026.

\bibitem{intro8-5}
J. Lee, H. Chung, Y. Cho, S. Kim and S. Hong, ``Near-Field Channel Estimation for XL-RIS Assisted Multi-User XL-MIMO Systems: Hybrid Beamforming Architectures," in IEEE Transactions on Communications, vol. 73, no. 3, pp. 1560-1574, March 2025.

\bibitem{intro8-6}
P. Zheng, X. Lyu, Y. Wang and Y. Gong, ``Convolutional Dictionary Learning-Based Hybrid-Field Channel Estimation for XL-RIS-Aided Massive MIMO Systems," in IEEE Transactions on Wireless Communications, vol. 24, no. 11, pp. 9085-9098, Nov. 2025.

\bibitem{intro8-7}
A. S. Gharagezlou, M. Rasti, S. Ali, S. K. Taskooh and M. Latva-aho, ``XL-RIS Placement Strategies for Beam Focusing in Coexisting Near-Field and Far-Field mmWave Communications," 2025 IEEE Wireless Communications and Networking Conference (WCNC), Milan, Italy, 2025, pp. 1-6.


\bibitem{intro9-pre}
T. Ma, B. Qian, X. Qin, X. Zhang, L. X. Cai and H. Zhou, ``Resource Scheduling for High-Capacity Multicast Service in Ultra-Dense LEO Satellite Networks," in IEEE Transactions on Vehicular Technology, vol. 73, no. 2, pp. 2468-2481, Feb. 2024.

\bibitem{intro9}
N. Jindal and Z. -q. Luo, ``Capacity Limits of Multiple Antenna Multicast," 2006 IEEE International Symposium on Information Theory, Seattle, WA, USA, 2006, pp. 1841-1845.

\bibitem{intro10}
S. Y. Park and D. J. Love, ``Capacity Limits of Multiple Antenna Multicasting Using Antenna Subset Selection," in IEEE Transactions on Signal Processing, vol. 56, no. 6, pp. 2524-2534, June 2008.

\bibitem{intro11}
L. Du, S. Shao, G. Yang, J. Ma, Q. Liang and Y. Tang, ``Capacity Characterization for Reconfigurable Intelligent Surfaces Assisted Multiple-Antenna Multicast," in IEEE Transactions on Wireless Communications, vol. 20, no. 10, pp. 6940-6953, Oct. 2021.

\bibitem{intro12}
B. Zhao, C. Ouyang, X. Zhang and Y. Liu, ``Channel Capacity of Near-Field Line-of-Sight Multiuser Communications," in IEEE Transactions on Wireless Communications, vol. 24, no. 5, pp. 4392-4409, May 2025.

\bibitem{intro2-2}
M. Cui and L. Dai, ``Channel Estimation for Extremely Large-Scale MIMO: Far-Field or Near-Field?," in IEEE Transactions on Communications, vol. 70, no. 4, pp. 2663-2677, April 2022.

\bibitem{channelSV}
X. He, H. Xu, G. Zhang, Y. Xia and Z. Wang, ``RIS-Assisted Joint Near-Field Channel Estimation and Beamforming Design Based on Deep Unfolding," in IEEE Transactions on Vehicular Technology, vol. 75, no. 3, pp. 4411-4423, March 2026.

\bibitem{ref1}
X. Wei, L. Dai, Y. Zhao, G. Yu and X. Duan, ``Codebook design and beam training for extremely large-scale RIS: Far-field or near-field?," in China Communications, vol. 19, no. 6, pp. 193-204, June 2022.

\bibitem{ref3}
K. Zhi et al., ``Performance Analysis and Low-Complexity Design for XL-MIMO With Near-Field Spatial Non-Stationarities," in IEEE Journal on Selected Areas in Communications, vol. 42, no. 6, pp. 1656-1672, June 2024.

\bibitem{ref3-1}
H. Lu and Y. Zeng, ``Communicating With Extremely Large-Scale Array/Surface: Unified Modeling and Performance Analysis," in IEEE Transactions on Wireless Communications, vol. 21, no. 6, pp. 4039-4053, June 2022.

\bibitem{convex}
Boyd, S., Vandenberghe, L. (2004). Convex Optimization Cambridge: Cambridge University Press. 

\bibitem{alg1}
Y. Geng, T. Hiang Cheng, K. Zhong and K. Chan Teh, ``Unified Manifold Optimization for Double-IRS-Aided MIMO Communication," in IEEE Communications Letters, vol. 28, no. 7, pp. 1713-1717, July 2024.

\bibitem{stein1993harmonic}
E.~M.~Stein,
\emph{Harmonic Analysis: Real-Variable Methods, Orthogonality,
	and Oscillatory Integrals}.
Princeton, NJ, USA: Princeton Univ. Press, 1993.

\bibitem{ref4}
M. R. Akdeniz et al., ``Millimeter Wave Channel Modeling and Cellular Capacity Evaluation," in IEEE Journal on Selected Areas in Communications, vol. 32, no. 6, pp. 1164-1179, June 2014.

\bibitem{sca}
V. Kumar, R. Zhang, M. D. Renzo and L. -N. Tran, ``A Novel SCA-Based Method for Beamforming Optimization in IRS/RIS-Assisted MU-MISO Downlink," in IEEE Wireless Communications Letters, vol. 12, no. 2, pp. 297-301, Feb. 2023.

\bibitem{cvx}
	Michael Grant and Stephen Boyd (2014). CVX: Matlab software for disciplined convex programming, version 2.1.






\end{thebibliography}
\end{document}